\documentclass{llncs}
\usepackage{amssymb,amsmath,graphics,color,upgreek,times,enumerate}
\usepackage{array}  
\usepackage{arydshln}
\usepackage{ragged2e} 
\usepackage{graphicx}
\newcommand{\NP}{{\sf NP}}
\newcommand{\cP}{{\sf P}}
\newcommand{\FPT}{{\sf FPT}}
\newcommand{\XP}{{\sf XP}}
\newcommand{\W}{{\sf W}}
\newcommand{\ssi}{\subseteq_i}
\newcommand{\si}{\supseteq_i}

\newtheorem{observation}{Observation}
\newtheorem{open}{Open Problem}

\def\figurename{Figure}

\begin{document}

\title{A Survey on the Computational Complexity of\\ Colouring Graphs with Forbidden Subgraphs
\thanks{The research leading to these results has received funding from EPSRC (EP/G043434/1) and the European Research Council under the European Union's Seventh Framework Programme (FP/2007-2013) / ERC Grant Agreement no. 267959.}}
\author{Petr A. Golovach\inst{1} \and Matthew Johnson\inst{2} \and Dani\"el Paulusma\inst{2} \and Jian Song\inst{2} 
}
\institute{
Department of Informatics, Bergen University,\\ PB 7803, 5020 Bergen, Norway\\
\texttt{petr.golovach@ii.uib.no}
\and
School of Engineering and  Computing Sciences, Durham University,\\
Science Laboratories, South Road,
Durham DH1 3LE, United Kingdom
\texttt{\{matthew.johnson2,daniel.paulusma,jian.song\}@durham.ac.uk}}
\maketitle

\begin{abstract}
For a positive integer $k$, a {\it $k$-colouring} of a graph $G=(V,E)$ is a mapping $c: V\rightarrow\{1,2,\ldots,k\}$ such that $c(u)\neq c(v)$ whenever $uv\in E$.  
The {\sc Colouring} problem is to decide, for a given $G$ and $k$, whether a $k$-colouring of $G$ exists.  If $k$ is fixed (that is, it is not part of the input), we have the decision problem $k$-{\sc Colouring} instead. 
We survey known results on the computational complexity of {\sc Colouring} and $k$-{\sc Colouring} for graph classes that are characterized by one or two forbidden induced subgraphs.  We also consider a number of variants: for example, where the problem is to extend a partial colouring, or where lists of permissible colours are given for each vertex.
Finally, we also survey results for graph classes defined by some other forbidden pattern.
\end{abstract}

\section{Introduction}\label{sec:intro}
To colour a graph is to label its vertices so that no two adjacent vertices have the same label.  We call the labels {\it colours}.  In a graph colouring problem one typically seeks to colour a graph using as few colours as possible, or perhaps simply to decide whether a given number of colours is sufficient.
Graph colouring problems are central to the study of both structural and algorithmic graph theory and have very many theoretical and practical applications.   Many variants and generalizations 
of the concept have been investigated, and there are some excellent 
surveys~\cite{Al93,KTV99,RS04b,Tu97} and a book~\cite{JT95} on the subject. 
We survey {\it computational complexity} results 
of graph colouring problems 
(for a short survey see~\cite{C14}).

As we will note in the following subsection, the complexity of many graph colouring problems is fully understood when the possible input is any graph, and it is therefore natural to study the complexity of problems where the input is restricted.  For example, one well-known result for graph colouring is due to Gr\"otschel, Lov\'asz, and Schrijver~\cite{GLS84} who have shown
that the problem of whether a \emph{perfect} graph can be coloured with at most $k$ colours for a given integer $k$ is polynomial-time solvable; in contrast, the problem for general graphs is \NP-complete~\cite{Ka72}.   

Perfect graphs are an example of a graph class that is closed under vertex deletion, and, like all such graph classes, can be characterized by a family of forbidden induced subgraphs (an infinite family in the case of perfect graphs).  In recent years, colouring problems for classes with forbidden-induced-subgraph characterizations have been extensively studied, and this survey is a response to the need for these results to be collected together.  In fact, such a task is beyond the scope of a single paper and so our aim here is to report on the computational complexity of graph colouring problems for graph classes characterized by the absence of {\it one or two} forbidden induced subgraphs 
(for a survey on computational complexity  results and open problems for colouring graphs characterized by more than two forbidden induced subgraphs or for which some graph parameter is bounded, see~\cite{Pa15}).

\subsection{Graph Colouring Problems}

We consider finite undirected graphs with no multiple edges and no self-loops.  That is, a graph
$G$ is an ordered pair $(V,E)$ that consists of a finite set $V$ of elements called {\it vertices} and a finite set $E$ of unordered pairs of members of $V$ called {\it edges}.    
 The sets $V$ and $E$ are called the {\it vertex set} and {\it edge set} of $G$, respectively, and an edge containing $u$ and $v$ is denoted $uv$.
The vertex and edge sets of a graph $G$ can also always be referred to as $V(G)$ and $E(G)$, and, when there is no possible ambiguity, we shall not always be careful in distinguishing between a graph and its vertex or edge set; that is, for example, we will write that a vertex belongs to a graph (rather than to the vertex set of the graph).
A graph $G'=(V',E')$ is a {\em subgraph} of $G$ (and $G$ is a {\em supergraph} of $G'$)
if $V'\subseteq V$ and $E'\subseteq E$; 
we say that $G'$ is a {\em proper} subgraph of $G$ if
$G'$ is a subgraph of $G$ and $G'\neq G$.

A {\em colouring} of a graph $G=(V,E)$ is a mapping $c: V\rightarrow\{1,2,\ldots \}$ such that $c(u)\neq c(v)$ whenever
$uv\in E$. We call $c(u)$  the {\it colour} of $u$. 
We let $c(U)=\{c(u)\; |\; u\in U\}$ for $U\subseteq V$.
If $c(V)\subseteq \{1,\ldots,k\}$, then $c$ is also called a {\em $k$-colouring} of $G$.
For a colour $c$, the set of all vertices of~$G$ with colour $c$ forms a 
\emph{colour class}.
We say that $G$ is $k$-colourable if a $k$-colouring exists, and the {\it chromatic number} of $G$ is the  smallest integer $k$ for which $G$ is $k$-colourable and is denoted $\chi(G)$.
A graph $G$ is \emph{$k$-vertex-critical} if $\chi(G)=k$ and $\chi(G')\leq k-1$ for any subgraph $G'$ of $G$ obtained by deleting a vertex.

We shall define a number of decision problems.

\bigskip
\noindent {\bf Colouring Problems}

\bigskip
\noindent \textsc{Colouring}\\
\mbox{}\rlap{\textit{Instance}\,: }\hphantom{\textit{Question}\,: }A graph $G$ and a positive integer $k$.\\
\textit{Question}\,: Is~$G$ $k$-colourable?

\bigskip\noindent
If $k$ is {\it fixed}, that is, not part of the input, then we have the following problem.

\bigskip
\noindent \textsc{$k$-Colouring}\\
\mbox{}\rlap{\textit{Instance}\,: }\hphantom{\textit{Question}\,: }A graph $G$.\\
\textit{Question}\,: Is~$G$ $k$-colourable?

\bigskip

\noindent {\bf Precolouring Extension Problems}

\bigskip

\noindent
A {\it $k$-precolouring} of a graph $G=(V,E)$ is a mapping $c_W :W\rightarrow\{1,2,\ldots k\}$ for some subset $W\subseteq V$. 
A $k$-colouring $c$ of $G$ is an extension of a $k$-precolouring $c_W$ of $G$ if $c(v)=c_W(v)$ for each $v \in W$.

\bigskip
\noindent \textsc{Precolouring Extension}\\
\mbox{}\rlap{\textit{Instance}\,: }\hphantom{\textit{Question}\,: }A graph $G$, a positive integer $k$ and a $k$-precolouring $c_W$ of $G$.\\
\textit{Question}\,: Can $c_W$ be extended to a $k$-colouring of $G$?

\bigskip\noindent
\noindent \textsc{$k$-Precolouring Extension}\\
\mbox{}\rlap{\textit{Instance}\,: }\hphantom{\textit{Question}\,: }A graph $G$ and a $k$-precolouring $c_W$ of $G$.\\
\textit{Question}\,: Can $c_W$ be extended to a $k$-colouring of $G$?

\bigskip

\noindent {\bf List Colouring Problems}

\bigskip
 
 \noindent
A {\it list assignment} of a graph $G=(V,E)$ is a function $L$ with domain $V$ such that  
for each vertex $u\in V$, $L(u)$ is a subset of $\{1, 2, \dots \}$. We refer to this set as the {\it list} of {\it admissible} colours for $u$.
If $L(u)\subseteq \{1,\ldots,k\}$ for each $u\in V$, then ${L}$ is also called a \emph{$k$-list assignment}.
 The {\it size} of a list assignment ${L}$ is the maximum list size $|L(u)|$ over all vertices $u\in V$.
A colouring $c$ 
{\it respects} ${L}$ if  $c(u)\in L(u)$ for all $u\in V$.
There are three decision problems as we can fix either the number of colours or the size of the list assignment. 

\bigskip
\noindent \textsc{List Colouring}\\
\mbox{}\rlap{\textit{Instance}\,: }\hphantom{\textit{Question}\,: }A graph $G$ and a list assignment $L$ for $G$.\\
\textit{Question}\,: Is there a colouring of $G$ that respects $L$?

\bigskip\noindent
\noindent \textsc{$\ell$-List Colouring}\\
\mbox{}\rlap{\textit{Instance}\,: }\hphantom{\textit{Question}\,: }A graph $G$ and a list assignment $L$ for $G$ of size at most $\ell$.\\
\textit{Question}\,: Is there a colouring of $G$ that respects $L$?

\bigskip\noindent
\noindent \textsc{List $k$-Colouring}\\
\mbox{}\rlap{\textit{Instance}\,: }\hphantom{\textit{Question}\,: }A graph $G$ and a $k$-list assignment $L$ for $G$.\\
\textit{Question}\,: Is there a colouring of $G$ that respects $L$?

\bigskip

Note that $k$-{\sc Colouring} can be viewed as a special case of $k$-{\sc Precolouring Extension} by choosing $W=\emptyset$, and that 
$k$-{\sc Precolouring Extension} can be viewed
as a special case of {\sc List $k$-Colouring} by choosing $L(u)=\{c_W(u)\}$ if $u\in W$ and $L(u)=\{1,\ldots,k\}$ if $u\in V\setminus W$.
Also {\sc List $k$-Colouring} can be readily seen to be a special case of {\sc $k$-List Colouring}, since if each list is a subset of $\{1, \ldots, k\}$, then the size of the list assignment is certainly at most $k$. 
Similarly, from our definitions, we see that it follows that, whenever $\ell_1\leq \ell_2$, $\ell_1$-{\sc List Colouring} is a special case of $\ell_2$-{\sc List Colouring}, 
and that whenever $k_1\leq k_2$,
{\sc List $k_1$-Colouring} is a special case of {\sc List $k_2$-Colouring}. 
 In Figure~\ref{f-col} we display all these relationships, which are implicitly assumed throughout the survey.
Having this figure in mind we can say that \NP-completeness results propagate upwards and polynomial time solvability results propagate downwards.
 Note that the relationships displayed in Figure~\ref{f-col} 
 remain valid even if we restrict
our attention  to special graph classes --- that is, if each of the problems accepts as input only certain graphs. 

Contrary to the list colouring variants,  when $\ell\geq k$, {\sc $k$-colouring} is not a special case of {\sc $\ell$-colouring}.
This is not only clear from its definition (the input consists of the graph only) but can also be illustrated by considering special graph classes. 
For example, {\sc 3-Colouring} is \NP-complete for planar graphs~\cite{GJS74}, whereas {\sc 4-Colouring} is polynomial time solvable for these graphs (since, of course, they are all 4-colourable)~\cite{AH89}.  Similarly,  {\sc $k$-Precolouring Extension} is not a special case of {\sc $\ell$-Precolouring Extension}.

\begin{figure}[ht]
\centering\scalebox{0.7}{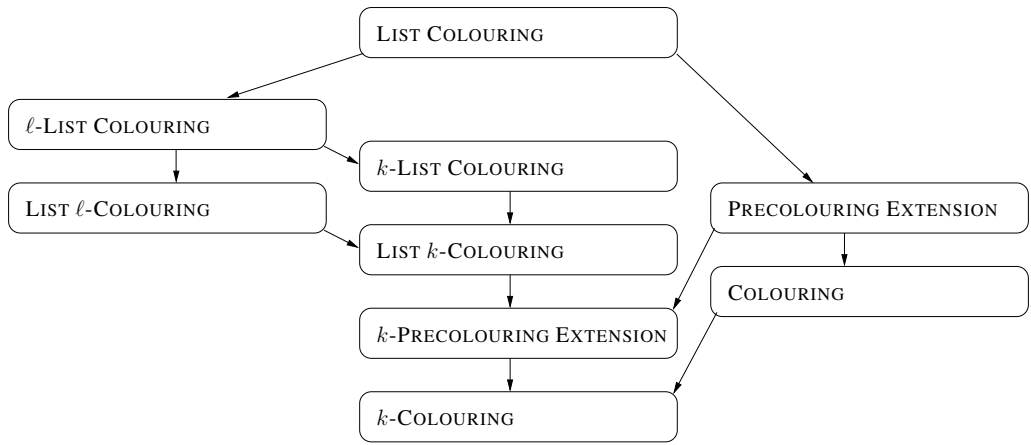}
\caption{Relationships between {\sc Colouring} and its variants. An arrow from one problem to another indicates that the latter is a special case of the former;  $k$ and $\ell$ are any two integers for which $\ell\geq k$. 
\label{f-col}}
\end{figure}

There is one further type of problem.

\medskip
\noindent {\bf Choosability Problems}

\bigskip

\noindent A graph $G=(V,E)$ is $\ell$-{\em choosable} if, for every list assignment ${L}$ of $G$ with $|L(u)|=\ell$ for all $u\in V$,
there exists a colouring that respects ${L}$.  

\bigskip\noindent
\noindent \textsc{Choosability}\\
\mbox{}\rlap{\textit{Instance}\,: }\hphantom{\textit{Question}\,: }A graph $G=(V,E)$ and a positive integer $\ell$.\\
\textit{Question}\,:  Is $G$ $\ell$-choosable?

\bigskip\noindent
\noindent \textsc{$\ell$-Choosability}\\
\mbox{}\rlap{\textit{Instance}\,: }\hphantom{\textit{Question}\,: }A graph $G=(V,E)$.\\
\textit{Question}\,:  Is $G$ $\ell$-choosable?

\bigskip\noindent
Theorem~\ref{t-general} describes the computational complexity of the  problems we have introduced on general graphs.
Here, $\Uppi_2^p$ is a complexity class in the polynomial hierarchy containing both \NP\ and co\NP; see for example the book of Garey and Johnson~\cite{GJ79} for its exact definition. 

\begin{theorem}\label{t-general}
The following two statements hold for general graphs.
\begin{itemize}
\item [(i)]  The problems $k$-{\sc Colouring}, $k$-{\sc Precolouring Extension}, {\sc List $k$-Colouring} and $k$-{\sc List Colouring} are
polynomial-time solvable if $k\leq 2$ and \NP-complete
if $k\geq 3$.\\[-8pt]
\item [(ii)] $\ell$-{\sc Choosability} is  polynomial-time solvable if $\ell\leq 2$ and $\Uppi_2^p$-complete if $\ell\geq 3$.
\end{itemize}
\end{theorem}

\begin{proof}
Lov\'asz~\cite{Lo73} showed that $3$-{\sc Colouring} 
is \NP-complete;
a straightforward reduction from $3$-{\sc Colouring} shows that $k$-{\sc Colouring} is \NP-complete for all $k\geq 4$.
Erd\"os,  Rubin and Taylor~\cite{ERT79} and Vizing~\cite{Vi79} observed that $2$-{\sc List Colouring} is polynomial-time solvable on general 
graphs. Then (i) follows from the relationships displayed in Figure~\ref{f-col}.
Erd\"os, Rubin and Taylor~\cite{ERT79} proved~(ii).\qed
\end{proof}
When considering Theorem~\ref{t-general}, a natural question to ask is whether further tractable cases can be found if restrictions are placed on the input graphs.  This survey reports progress on finding answers to this question.

\subsection{Notation and Terminology}\label{s-known}

We define the graph classes considered in this survey and other notation and terminology. We refer to the textbook of Diestel~\cite{Di05} for any undefined terms.

Let $G=(V,E)$ be a graph.
For a subset $S\subseteq V$, let $G[S]$ denote the {\it induced} subgraph of $G$ that has vertex set~$S$ and edge set $\{uv\in E(G)\; |\; u,v\in S\}$. 
For a subset $S\subseteq V$, we write $G-S=G[V\setminus S]$, and for a vertex $v\in V$, we use $G-v=G-\{v\}$.
For a graph $F$, we write $F\subseteq G$ and $F\ssi G$ to denote that $F$ is a subgraph or an  induced subgraph of $G$, respectively.
For two graphs $G$ and $H$, a vertex mapping $f:V(G)\to V(H)$ is called a ({\it graph}) {\it isomorphism} when $uv\in E(G)$ if and only if $f(u)f(v)\in E(H)$, and we say that $G$ and $H$ are {\it isomorphic} whenever such a mapping exists.
Let $G$ be a graph and $\{H_1,\ldots,H_p\}$ be a set of graphs.  Then
$G$ is {\it $(H_1,\ldots,H_p)$-free} if $G$ has no \emph{induced} subgraph isomorphic to a graph in $\{H_1,\ldots,H_p\}$.
And $G$ is {\it strongly $(H_1,\ldots,H_p)$-free} if $G$ has no subgraph isomorphic to a graph in $\{H_1,\ldots,H_p\}$.   If $p=1$, we can simply write that $G$ is (strongly) $H_1$-free (rather than (strongly) $(H_1)$-free).
 
\begin{observation}\label{o-in}
If a graph $H'$ is an induced subgraph of a graph $H$, then  every $H'$-free graph is $H$-free.
If $H'$ is a subgraph of $H$, then every strongly $H'$-free graph is strongly $H$-free.
\end{observation}

The {\it complement} of a graph $G$ is denoted $\overline{G}$ and has the same vertex set as $G$ and an edge between two distinct vertices 
if and only if these vertices are not adjacent in $G$.
The {\it union} of two graphs $G$ and $H$ is the graph with vertex set $V(G)\cup V(H)$ and edge set $E(G)\cup E(H)$. If $V(G)\cap V(H)=\emptyset$, then we call the union of $G$ and $H$ the {\it disjoint union} of $G$ and $H$ and denote it $G+H$.
We denote the disjoint union of $r$ copies of $G$ by $rG$.

For a graph $G$, the {\em degree} $\deg_G(u)$ of a vertex $u$ in $G$ is the number
of edges incident with it, or equivalently the size of its {\it neighbourhood} $N_G(u)=\{v\in V\; |\; uv\}$.
A vertex $u$ that is adjacent to all other vertices of $G$ is called a \emph{dominating} vertex of $G$.
The {\it minimum degree} of $G$  is the smallest  degree of a vertex in $G$, and the {\it maximum degree} of $G$, denoted by $\Delta(G)$,  is the largest degree of a vertex in $G$. 
If every vertex in $G$ has degree $p$,  then $G$ is said to be {\it $p$-regular} (or sometimes just regular).

\begin{figure}[ht]
\centering\scalebox{0.65}{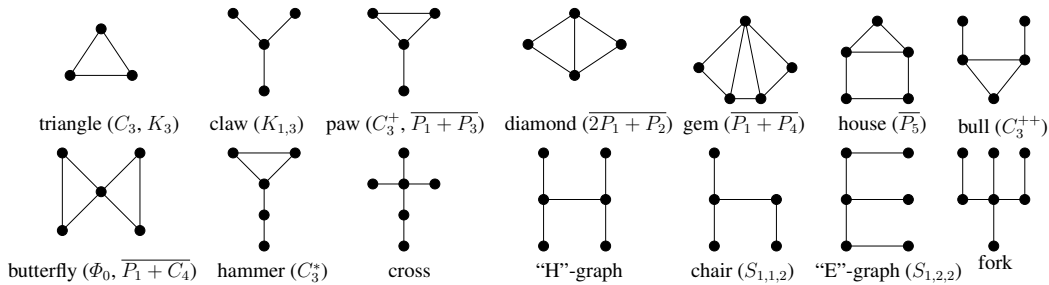}
\caption{A number of small graphs with special names that we use throughout the survey.  Also indicated are notations that will be defined in later sections. 
\label{f-small}}
\end{figure}

For $n\geq 1$,  the {\it complete graph} $K_n$ is a graph on $n$ vertices in which each pair of distinct vertices is joined by an edge.  For a graph $G$, a subgraph isomorphic to a complete graph is called a \emph{clique}, and the {\it clique number} of $G$ is the size of its largest clique and is denoted~$\omega(G)$.

For $n\geq 1$, the graph with vertices $\{u_1,\ldots,u_n\}$ and edges $\{u_1u_2, u_2u_3, \ldots, u_{n-1}u_n\}$ is called a \emph{path} and is denoted $P_n$.  For $n\geq3$, the graph obtained from $P_n$ by adding the edge $u_1u_n$, is called a cycle and is denoted $C_n$. 
The {\it length} of a path or cycle is its number of edges. 
The {\em end-vertices} of a path are the vertices of degree 1 (we will also refer to the vertices that comprise an edge as its end-vertices).
The graph $C_3=K_3$ is also called a {\it triangle} (see Figure~\ref{f-small}), 
and a $C_3$-free graph is also called {\it triangle-free}.   A $P_4$-free graph is also called a {\it cograph}.
Notice that $rP_1$ denotes an {\it independent set} on $r$ vertices. 

Let $G=(V,E)$ be a graph.
The {\it girth} of  $G$ is the length of a shortest cycle in~$G$ or infinite if~$G$ has no cycle. 
Note that a graph has girth at least~$g$ for some integer $g\geq 4$ if and only if it is $(C_3,\ldots,C_{g-1})$-free.
We say that $G$ is  {\em connected} if there is a path between every pair of distinct vertices; otherwise it is called \emph{disconnected}.
A vertex $u\in V$ is a \emph{cut vertex} if $G$ is connected and~$G-u$ is disconnected.
If $G$ is connected and has no cut vertices, it is \emph{$2$-connected}. 
A maximal connected subgraph of $G$ is called a {\em connected component}.
A graph is a {\it tree} if it is connected and $(C_3,C_4,\ldots)$-free. A graph is a {\it forest} if each of its connected components is a tree.
A graph is a {\it linear forest} if each of its connected components is a path.

A graph is {\it bipartite} if its vertex set can be partitioned into two sets
such that every edge has one end-vertex in each set.
For $r\geq 1$, $s\geq 1$, the {\it complete bipartite graph} $K_{r,s}$ is a bipartite graph whose vertex set can be partitioned into two sets of sizes $r$ and $s$ such that there is an edge joining each pair of vertices from distinct sets.
For $r\geq 1$, the graph $K_{1,r}$ is also called a {\it star}.
The graph~$K_{1,3}$ is also called a {\it claw}
(see Figure~\ref{f-small}), 
and a $K_{1,3}$-free graph is called {\it claw-free}.
A graph is a {\it complete multipartite} graph if  the vertex set can be partitioned so that there is an edge joining every pair of vertices from distinct sets of the partition and no edge joining vertices in the same set.

A graph~$G$ is {\em perfect} if, for every induced subgraph $H\ssi G$, the chromatic number of~$H$ equals its clique number.
The {\it line graph} of a graph $G=(V,E)$ has vertex set $E$ and $x, y \in E$ are adjacent as vertices in the line graph 
if and only if
they are adjacent as edges in $G$; that is, 
if they share an end-vertex in $G$.
A graph is \emph{planar} if it can be drawn in the plane so that its edges intersect only at their end-vertices. 
A graph is a {\it split} graph if its vertices can be partitioned into two sets that induce a clique and an independent set; if every vertex in the independent set is adjacent to every vertex in the clique, then it is a {\it complete} split graph.
A number of small graphs that have special names are shown in Figure~\ref{f-small}. 

A \emph{tree decomposition} of a graph $G$ is a tree $T$
where the elements of $V(T)$ (called \emph{nodes}) are subsets of $V(G)$ such that the following three conditions are satisfied:
\begin{itemize} 
\item for each vertex $v\in V(G)$, there is 
at least one node $X\in V(T)$ with $v\in X$; \item for each edge $uv\in E(G)$, there is a node $X\in V(T)$ with $\{u,v\}\subseteq X$; 
\item for each vertex $v\in V(G)$, the set of nodes $\{X \mid v\in X\}$ 
induces a connected subtree of $T$.  
\end{itemize}
If $X$ is the largest node in a tree decomposition, then the \emph{width} of the decomposition is $|X|-1$. The {\it treewidth} of $G$ is the minimum width over all possible tree decompositions of $G$.
If a tree decomposition $T$ is a path, then it is a \emph{path decomposition}. The {\it pathwidth} of $G$ is the minimum width over all possible path decompositions of~$G$.

The graph parameter \emph{clique-width} is defined by considering how to construct graphs in which each vertex has a label.  Four operations are permitted:
\begin{itemize}
\item create a graph with one (labelled) vertex;
\item combine two labelled graphs by taking their disjoint union;
\item in a labelled graph, for two labels $i$ and 
$j$ with $i\neq j$, join by an edge each vertex with label $i$ to each vertex with label~$j$; 
\item in a labelled graph, for two labels $i$ and $j$, change every instance of label $i$ to $j$.
\end{itemize}
The {\em clique-width} of  $G$ is the minimum  number of labels needed to
construct $G$ (with some labelling) using these operations.
A description of how $G$ is constructed using these operations is called a \emph{$q$-expression} if $q$ is the number of labels used
(so the clique-width of $G$ is the minimum $q$ for which $G$ has a $q$-expression).
We say that  a class of graphs ${\cal G}$ has \emph{bounded} clique-width (or bounded treewidth) if 
there is a constant $p$ such that the clique-width (or treewidth) of every graph in ${\cal G}$ is at most $p$. 

Let $G=(V,E)$ be a graph.
The \emph{contraction} of an edge $uv\in E$ removes $u$ and $v$ from $G$, and adds a new vertex $w$ and edges such that the neighbourhood of $w$ is the union of the neighbourhoods of $u$ and $v$. 
Note that, by definition, edge contractions create neither self-loops nor multiple edges.
Let $u\in V$ be a vertex of degree 2 whose neighbours $v$ and $w$  are not adjacent.
The \emph{vertex dissolution} of $u$ removes $u$ and adds the edge $vw$.
The ``dual'' operation of a vertex dissolution is  \emph{edge subdivision}, which replaces an edge $vw$  by a new vertex $u$ and edges $uv$ and $uw$. 
We say that $G$ contains another graph $H$ as a {\it minor}  if $G$ can be modified into $H$ by 
a sequence that consists of edge contractions, edge deletions  and vertex deletions.
And $G$ contains $H$ as a {\it topological minor} if $G$ can be modified into $H$ by 
a sequence that consists of vertex dissolutions, edge deletions and vertex deletions.

\section{Results and Open Problems for $H$-Free Graphs}\label{s-hfree}

In this section we consider graph classes characterized by one forbidden induced subgraph; we refer to the collection of all such graph classes as $H$-free graphs.
In Section~\ref{ss-col} we consider Colouring, Precolouring Extension and List Colouring Problems,
and in Section~\ref{ss-choos} we consider Choosability Problems.

\subsection{Colouring, Precolouring Extension and List Colouring Problems}\label{ss-col}

Theorem~\ref{t-kktw} below describes what is known about the complexity of problems where the number of colours is not fixed.  We first briefly describe the origin of these results.
 
Kr\'al', Kratochv\'{\i}l, Tuza, and Woeginger~\cite{KKTW01} completely classified the computational complexity of {\sc Colouring}  by showing that 
it is polynomial-time solvable for $H$-free graphs if $H$ is an induced subgraph of 
$P_4$ or of $P_1+P_3$, and \NP-complete otherwise.
Both Hujter and Tuza~\cite{HT96} and
Jansen and Scheffler~\cite{JS97} showed that {\sc Precolouring Extension} is polynomial-time solvable for $P_4$-free graphs.
This result was used by Golovach, Paulusma and Song~\cite{GPS12} in order to obtain a dichotomy for {\sc Precolouring Extension} analogous to the one
of Kr\'al' et al. 
Jansen and Schefller~\cite{JS97} also showed the following result which we 
state as a Theorem as we will use it later in the paper.

\begin{theorem}
\label{t-cb}
$3$-{\sc List Colouring}  is \NP-complete for complete bipartite graphs.
\end{theorem}
As a consequence, $3$-{\sc List Colouring} is \NP-complete for $(P_1+P_2)$-free graphs. 
Jansen~\cite{Ja96} implicitly showed that 3-{\sc List Colouring} is \NP-complete for
(not necessarily vertex-disjoint) unions of two complete graphs, and thus for $3P_1$-free graphs.
By combining these results, together with Theorem~\ref{t-general}~(i),
Golovach et al.~\cite{GPS12} obtained dichotomies for  {\sc List Colouring} and $\ell$-{\sc List Colouring}. We summarize all these results:

\begin{theorem}\label{t-kktw}
Let $H$ be a graph. Then the following four statements hold for $H$-free graphs.
\begin{itemize}
\item [(i)] {\sc Colouring} is polynomial-time solvable if $H$ is an induced subgraph of 
$P_4$ or of $P_1+P_3$; otherwise it is \NP-complete.\\[-9pt]
\item [(ii)] {\sc Precolouring Extension} is polynomial-time solvable if $H$ is an induced subgraph of 
$P_4$ or of $P_1+P_3$; otherwise it is \NP-complete.\\[-9pt]
\item [(iii)] {\sc List Colouring} is polynomial-time solvable if $H$ is an induced subgraph of $P_3$; otherwise
it is \NP-complete.\\[-9pt]
\item [(iv)]  
For $\ell\geq 3$, {\sc $\ell$-List Colouring} is polynomial-time solvable if $H$ is an induced subgraph of $P_3$;
otherwise it is \NP-complete.
[Recall that for  $\ell\leq 2$, {\sc $\ell$-List Colouring} is polynomial-time solvable on general graphs.]\\[-9pt]
\end{itemize}
\end{theorem}

Theorem~\ref{t-kktw} gives a complete complexity classification for problems where the number of colours is not fixed; that is, it is part of the input.   
Once such a classification was found, the natural direction for further research was to impose 
an upper bound on the number of available colours, and there is now an extensive  literature on such problems.  We survey the known results. 
We start, in Theorems~\ref{t-kl07} and~\ref{t-line}, with 
more general results; we will soon see why they are useful.

Kr\'al' et al.~\cite{KKTW01} 
showed (in order to prove that {\sc Coloring} is \NP-complete for $H$-free graphs whenever $H$ has a cycle)
that $3$-{\sc Colouring} is \NP-complete for graphs of girth at least $g$ for any fixed $g\geq 3$.
Using a similar reduction, Kami\'nski and Lozin~\cite{KL07} extended this result to all $k\geq 3$ 
though in fact a stronger result had been previously obtained by Emden-Weinert, Hougardy and Kreuter~\cite{EHK98}:

 \begin{theorem}\label{t-kl07}
For all $k\geq 3$ and all $g\geq 3$, {\sc $k$-Colouring}  is \NP-complete for graphs with girth at least~$g$ and with maximum degree at most~$6k^{13}$ . 
\end{theorem}
Theorem~\ref{t-kl07} implies that for any $k\geq 3$,  $k$-{\sc Colouring} is \NP-complete for the class of $H$-free graphs 
whenever $H$ contains a cycle.  Let us remind the reader once more
that Figure~\ref{f-col} tells us that \NP-completeness results propagate upwards, 
which, combined with Theorems~\ref{t-kktw} and~\ref{t-kl07}, allows us to say that the complexity of Colouring, Precolouring Extension and List Colouring problems for $H$-free graphs is classified except when $H$ is a forest.

The following theorem is due to Holyer~\cite{Ho81}, who settled the case $k=3$,  and Leven and Galil~\cite{LG83} who settled the case $k\geq 4$.

\begin{theorem}\label{t-line}
For all $k\geq 3$, $k$-{\sc Colouring} is \NP-complete for line graphs of $k$-regular graphs.
\end{theorem}

Because line graphs are easily seen to be claw-free, Theorem~\ref{t-line} implies that  for all $k\geq 3$, $k$-{\sc Colouring} is \NP-complete on $H$-free graphs whenever $H$ is a forest with a vertex of degree at least $3$. This leaves only the case in which $H$ is a linear forest. 
 
Combining a result from Balas and Yu~\cite{BY89}
on the number of maximal independent sets in an $sP_2$-free graph 
and a result from Tsukiyama, Ide, Ariyoshi and Shirakawa~\cite{TIAS77} on the enumeration of such sets leads to the result that {\sc  $k$-Colouring} is polynomial-time solvable
on $sP_2$-free graphs for any two integers $k$ and~$s$; 
see, for example, the paper of Dabrowski, Lozin, Raman and Ries~\cite{DLRR12} for a proof of this result.
By a few additional arguments, it is possible to obtain the following 
new result, which is stronger (notice that polynomial-time results propagate downwards in Figure~\ref{f-col}). 

\begin{theorem}\label{t-sp2}
For all $k\geq 1$, $s\geq 1$,  {\sc List $k$-Colouring}  is polynomial-time solvable on $sP_2$-free graphs.
\end{theorem}

\begin{proof}
Let $k\geq 1$ and $s\geq 1$. Let $G$ be an $sP_2$-free graph with a $k$-list assignment $L$.
By the results of Balas and Yu~\cite{BY89} and Tsukiyama et al.~\cite{TIAS77},
we can enumerate all maximal independent sets of $G$ in 
polynomial time.
For each maximal independent set $I$ and each colour $i\in\{1,\ldots,k\}$, 
we colour each vertex of $W=\{u \in I : i \in L(u) \}$ with $i$, and then, recursively, attempt to colour $G-W$ with the remaining colours. 
The running time of this algorithm is $(kn)^{O(k)}$.
The algorithm can fail: it might not colour every vertex. However, if it succeeds then the resulting colouring will respect $L$. 

It remains to show that the algorithm will find a colouring if one exists. 
Consider the set of vertices~$W$ coloured $i$ in some colouring.  They belong to a maximal independent set $I$, and we can assume that $W=\{u \in I : i \in L(u) \}$ (by changing the colours of some vertices if necessary; the colouring will still be proper).  So at some point the algorithm will consider $i$ and $I$ and colour $W$ with $i$.  By applying the same argument to $G-W$ (which we know can be coloured with the remaining colours), we can see that the algorithm will obtain a colouring.\qed
\end{proof}

The following theorem summarizes what is known for colouring problems on $H$-free graphs when the number of colours is fixed.

\begin{theorem}\label{t-polysum}
Let $H$ be a graph. Then the following five statements hold:
\begin{itemize}
\item[(i)] {\sc $k$-Colouring} is \NP-complete for $H$-free graphs if 
\begin{enumerate}[{\em 1.}]
\item  $k\geq3$ and $H\si C_r$ for $r\geq 3$
\item  $k\geq 3$ and $H\si K_{1,3}$
\item  $k\geq 4$ and $H\si P_7$
\item  $k\geq 5$ and $H\si P_6$.\\[-8pt]
\end{enumerate}
\item[(ii)] {\sc List $k$-Colouring} is \NP-complete for $H$-free graphs if 
\begin{enumerate}[{\em 1.}]
\item $k\geq 4$ and $H\si P_6$
\item  $k\geq 5$ and $H\si P_2+P_4$.\\[-8pt]
\end{enumerate}
\item [(iii)] {\sc List $k$-Colouring} is polynomial-time solvable for $H$-free graphs if $k\leq 2$ or
\begin{enumerate}[{\em 1.}]
\item  $k\leq 3$ and $H\ssi sP_1+P_7$ for $s\geq 0$ 
\item $k\leq 3$ and $H\ssi sP_3$ for $s\geq 1$
\item  $k\geq 1$ and $H\ssi sP_1+P_5$ for  $s\geq 0$ 
\item $k\geq 1$ and $H\ssi sP_2$ for  $s\geq 1$.\\[-8pt]
\end{enumerate}
\item [(iv)] {\sc $4$-Precolouring Extension} is polynomial-time solvable for $H$-free graphs if 
$H\ssi P_2+P_3$.\\[-8pt]
\end{itemize}
\end{theorem}

\begin{proof}
For each case, we refer to the literature or to a result stated above.  In some cases we will make additional comments referring to earlier (weaker) results that provided techniques or suggested approaches that were important in obtaining the final result.
\begin{itemize}
\item [(i)] We first consider the \NP-completeness results for {\sc $k$-Colouring}.  
\begin{enumerate}
\item This follows immediately from Theorem~\ref{t-kl07}.
\item This is a direct consequence of Theorem~\ref{t-line} and the fact that every line graph is claw-free.

\item Woeginger and Sgall~\cite{WS01} showed that $4$-{\sc Colouring} is \NP-complete for $P_{12}$-free graphs. 
This bound was improved in a number of other papers.
First, Le, Randerath and Schiermeyer~\cite{LRS07} showed that 4-{\sc Colouring} is \NP-complete for $P_9$-free graphs.
Then, Broersma, Golovach, Paulusma and Song~\cite{BGPS12b} showed that $4$-{\sc Colouring} is \NP-complete for $P_8$-free graphs.
Finally, the strongest  \NP-completeness result for $4$-{\sc Colouring} is due to Huang~\cite{Hu13}, who showed that it is \NP-complete for $P_7$-free graphs (we note also that Broersma et al.~\cite{BGPS12b} had already shown that $4$-{\sc Precolouring Extension} is \NP-complete
for $P_7$-free graphs).

\item  Broersma et al.~\cite{BFGP13} had shown that $5$-{\sc Precolouring Extension} is \NP-complete for $P_6$-free graphs. Huang~\cite{Hu13} improved this (and also a result of Woeginger and Sgall~\cite{WS01} who showed
that $5$-{\sc Colouring} is \NP-complete for $P_8$-free graphs) by proving that $5$-{\sc Colouring} is \NP-complete for $P_6$-free graphs.
\end{enumerate}

\item [(ii)] Next we look at the \NP-completeness results for {\sc List-$k$-Colouring}.
\begin{enumerate}
\item  This is a result of
 Golovach, Paulusma and Song~\cite{GPS12}.

\item  Couturier, Golovach, Kratsch and Paulusma~\cite{CGKP11} showed that  {\sc List $k$-Colouring} is \NP-complete for some integer $k$ on $H$-free graphs, whenever $H$ is a supergraph of $P_1 + P_5$ with at least five edges. In particular, 
they proved that {\sc List 5-Colouring} is \NP-complete on $(P_2+P_4)$-free graphs.
\end{enumerate}
\item [(iii)] We now turn to the polynomial-time results for
{\sc List-$k$-Colouring}.\footnote{Two
of these results were only formulated in the literature for {\sc $k$-Precolouring Extension} instead of for {\sc List $k$-Colouring}. In the Appendix we give proofs for {\sc List $k$-Colouring} or explain how the known proofs 
for {\sc $k$-Precolouring Extension} can be modified accordingly.}
Before we consider the individual cases, we discuss an observation of  Broersma et al.~\cite{BGPS12b} that we will use twice.
They noticed that $3$-{\sc Precolouring Extension}  is polynomial-time solvable
for $(P_1+H)$-free graphs whenever it is polynomial-time solvable for $H$-free graphs (and by repeated application the problem is, in fact, solvable for $(sP_1+H)$-free graphs for any $s \geq 0$). We note that
analogous statements can be made about 3-{\sc Colouring} and {\sc List 3-Colouring}.
\begin{enumerate}
\item Randerath and Schiermeyer~\cite{RS04a} showed that $3$-\textsc{Colouring} is polynomial-time solvable on $P_6$-free graphs.
This was generalized by Broersma, Fomin, Golovach and Paulusma~\cite{BFGP13} who showed that 
{\sc 3-Precolouring Extension} is polynomial-time solvable for $P_6$-free graphs. In fact, their proof shows polynomial-time solvability of {\sc List 3-Colouring}
for $P_6$-free graphs. 
Broersma et al.~\cite{BGPS12b}  showed that {\sc 3-Precolouring Extension} can be solved in polynomial time  on $(P_2+P_4)$-free graphs.
Their proof can be used to show that {\sc List $3$-Colouring} is polynomial-time solvable on $(P_2+P_4)$-free graphs.
Recently, Bonomo, Chudnovsky, Maceli, Schaudt, Stein and Zhong~\cite{BCMSSZ} gave an $O(|V|^{23}|)$ time algorithm for {\sc List 3-Colouring} on $P_7$-free graphs (thereby solving Problem~17 in~\cite{RST02} and Problem~56 in~\cite{RS04b}).  The same authors showed that {\sc 3-Colouring} can be solved in $O(|V|^7|)$ time on $(C_3,P_7)$-free graphs.
\item A further result of Broersma et al.~\cite{BGPS12b} showed that
{\sc 3-Precolouring Extension}  is polynomial-time solvable on $sP_3$-free graphs for all $s\geq 1$. 
In fact, 
though they did not state it explicitly,
the result
holds for {\sc List $3$-Colouring} on $sP_3$-free graphs.
\item This is a result of Couturier et al.~\cite{CGKP11}. It generalizes an earlier result of  Ho\`ang, Kami\'nski, Lozin, Sawada, and Shu~\cite{HKLSS10} who proved that for every integer $k\geq 1$, {\sc List $k$-Colouring} is polynomial-time solvable on $P_5$-free graphs.
Previously, (different) proofs for the case $k\leq 3$ were given by Woeginger and Sgall~\cite{WS01} and Randerath, Schiermeyer and Tewes~\cite{RST02}. 
\item This is Theorem~\ref{t-sp2}.  
\end{enumerate}

\item [(iv)] This is a result of
Golovach, Paulusma and Song~\cite{GPS13}.\qed
\end{itemize}
\end{proof}
As a consequence of Theorem~\ref{t-polysum}
we obtain dichotomies for {\sc $k$-Colouring}, {\sc $k$-Precolouring Extension} and {\sc List $k$-Colouring} when $H$ is small. 
These are stated in Theorem~\ref{t-dicho}.  

\begin{theorem}\label{t-dicho}
Let $H$ be a graph and $k$ an integer.
Then the following three statements hold:
\begin{itemize}
\item [(i)] If $|V(H)|\leq 6$,
then $3$-{\sc Colouring}, $3$-{\sc Precolouring Extension}, {\sc List $3$-Colouring}
are polynomial-time solvable on $H$-free graphs if $H$ is a linear forest, and \NP-complete otherwise.\\[-8pt]
\item [(ii)]
If $|V(H)|\leq 5$,
then $4$-{\sc Colouring}, $4$-{\sc Precolouring Extension}
are polynomial-time solvable on $H$-free graphs if $H$ is a linear forest, and \NP-complete otherwise.\\ [-8pt]
\item [(iii)]
If $|V(H)|\leq 4$ and $k\geq 5$,
then $k$-{\sc Colouring}, $k$-{\sc Precolouring Extension}, {\sc List $k$-Colouring}
are polynomial-time solvable on $H$-free graphs if $H$ is a linear forest, and \NP-complete otherwise.
\end{itemize}
\end{theorem}
Note that statement (ii) of Theorem~\ref{t-dicho} cannot be stated also for {\sc List 4-Colouring} due to exactly one missing case, which is the complexity of {\sc List $4$-Colouring}
for $(P_2+P_3)$-free graphs. 

Theorem~\ref{t-polysum} also implies that for $H$-free graphs,
$3$-{\sc Colouring} is classified for all graphs $H$ on seven vertices except when $H\in \{P_2+P_5,P_3+P_4\}$, that
$4$-{\sc Colouring} is classified for all graphs $H$ on six vertices, except when $H\in 
\{P_1+P_2+P_3,P_2+P_4,2P_3,P_6\}$, and that
$5$-{\sc Colouring} is classified for all graphs $H$ on five vertices, except when $H=P_2+P_3$.

Table~\ref{t-table1} shows a summary of the existing results for $P_r$-free graphs obtained from Theorem~\ref{t-polysum}. We include this table, because {\sc $k$-Colouring} restricted to graphs characterized by forbidden induced subgraphs was most actively studied
for forbidden 
induced 
paths. By comparing Table~\ref{t-table1} with similar tables that can be found in several earlier papers~\cite{BGPS12b,GPS12,HKLSS10,LRS07,RS04a,RS04b,WS01} one can see the gradual progress that has been made over the years.

\begin{table}[h]
\begin{center}
\resizebox{360pt}{35pt}{
\begin{tabular}{c|c|c|c|c||c|c|c|c||c|c|c|c}
& \multicolumn{4}{c||}{{\sc $k$-Colouring}} & \multicolumn{4}{c||}{ {\sc $k$-Precolouring Extension}}&\multicolumn{4}{c}{ {\sc List $k$-Colouring}}\\
\hline
$\;\;r\;\;$                                 & {\small $k=3$}             & {\small $k=4$}                  & {\small $k=5$}               & {\small $k\ge 6$}                          & {\small $k=3$}             & {\small $k=4$}                  & {\small $k=5$}               & {\small $k\ge 6$}                             & {\small $k=3$}             & {\small $k=4$}                  & {\small $k=5$}               & {\small $k\ge 6$} \\
\hline
{\small $r\leq 5$}    & {\small P} & {\small P}      & {\small P} & {\small P}  & {\small P} & {\small P}      & {\small P} & {\small P} 
 & {\small P} & {\small P}      & {\small P} & {\small P} \\
{\small $r=6$}         & {\small P} & ?                     & {\small NP-c}               & {\small NP-c}                 & {\small P} & ?                     & {\small NP-c}            & {\small NP-c} 
&{\small P} & {\small NP-c}                  & {\small NP-c}            & {\small NP-c} \\
{\small $r=7$}         & 
{\small P}             
& {\small NP-c}                     & {\small NP-c}               & {\small NP-c}  & {\small P}              &  {\small NP-c}    & {\small NP-c}        & {\small NP-c}
& {\small P}              & {\small NP-c}    & {\small NP-c}        & {\small NP-c}\\
{\small $r\geq 8$}   & ?              &  {\small NP-c}    &{\small NP-c} & {\small NP-c}  &? &  {\small NP-c}    &{\small NP-c} & {\small NP-c} 
  &?              & {\small NP-c} & {\small NP-c} & {\small NP-c} 
\end{tabular}
}
\end{center}
\vspace*{0.3cm}
\caption{The complexity of $k$-{\sc Colouring}, $k$-{\sc Precolouring Extension}  and {\sc List $k$-Colouring}  
on $P_r$-free graphs for fixed $k$ and $r$.}\label{t-table1}
\end{table}

\begin{open}\label{o-class}
Complete the classification of the complexity of $k$-{\sc Colouring}, $k$-{\sc Precolouring Extension} and {\sc List $k$-Colouring} for $H$-free graphs.
\end{open}

\medskip
\noindent
{\bf Two Important Subproblems}
First, as noted, the complexity status of $4$-{\sc Colouring} for $P_6$-free graphs is still open.
One of the key ingredients in the proofs of the two aforementioned hardness results of $4$-{\sc Colouring} for $P_7$-free graphs
and $5$-{\sc Colouring} for $P_6$-free graphs by Huang~\cite{Hu13} are the so-called 
\emph{nice $k$-critical} graphs. A graph $G=(V,E)$ is nice $k$-critical for some integer $k$ if it is $k$-vertex-critical,
and if moreover, $G$ contains three independent vertices $v_1$, $v_2$, $v_3$
such that $\omega(G-\{v_1,v_2,v_3\})=\omega(G)=k-1$. In his hardness reductions, Huang~\cite{Hu13} uses the existence of 
$P_7$-free nice $3$-critical graphs and $P_6$-free nice $4$-critical graphs. He also proved  that $P_6$-free nice 3-critical graphs do not exist. Hence, new techniques are required to determine the computational complexity of $4$-{\sc Colouring} for $P_6$-free graphs.

The second intriguing open question (Problem~18 in~\cite{RST02} and Problem~57 in~\cite{RS04b}) that must be answered when solving Open Problem~\ref{o-class} is whether there exists an integer $r\geq 8$ such that $3$-{\sc Colouring} is \NP-complete for $P_r$-free
graphs. This is also unknown for  {\sc $3$-Precolouring Extension} and {\sc List $3$-Colouring}. As observed by Golovach et al.~\cite{GPS12}, an affirmative answer for one of the three problems leads to an affirmative answer for the other two.
We also note that there is no graph~$H$ and integer $k$ known  for which the computational complexity of the problems {\sc $k$-Colouring}, {\sc $k$-Precolouring Extension} and {\sc List $k$-Colouring}
differs for $H$-free graphs 
(whether such a graph $H$ exists was posed as an open problem by Huang, Johnson and Paulusma~\cite{HJP14}).

\subsubsection{Parameterized Complexity Theory}
Parameterized complexity theory is a framework  that offers a refined analysis of 
\NP-hard algorithmic problems. 
We measure the complexity of  a problem not only in terms of the input length but also in terms of a parameter, which is a numerical value not necessarily dependent on the input length.
The instance of a parameterized problem is  a pair $(I,p)$, where $I$ 
is the problem instance and $p$ is the parameter.
The choice of parameter will  depend on the structure of the problem (and there might be many possible choices).  

The central notion in parameterized complexity theory is the concept of \emph{fixed-parameter tractability}. 
A problem is called \emph{fixed-parameter tractable} (\FPT) if every instance $(I,p)$ can be solved in time $f(p)|I|^{O(1)}$ where $f$ is a computable function that only depends on $p$. The complexity class \FPT\ is the class of all fixed-parameter tractable problems.
The complexity class \XP\ is the class of all problems that can be solved in time 
$|I|^{f(p)}$.

By definition $\FPT \subseteq \XP$, but a collection of intermediate complexity classes has been defined as well.
It is known as the 
\W-hierarchy:
\[\FPT=\W[0]\subseteq \W[1]\subseteq \W[2]\subseteq \ldots \subseteq \W[P]\subseteq \XP.\] 
It is widely believed that $\FPT\neq \W[1]$. Hence, if a
problem is hard for some class $\W[i]$, then it is considered to be fixed-parameter intractable. 

A problem is \emph{para-\NP-complete}  when it is \NP-complete for some fixed value of the parameter.
Such a problem is not in \XP (and so not in \FPT) unless $\cP=\NP$.
We refer the reader to the textbook of  Niedermeier \cite{Niedermeier06} for further details.

For  {\sc Colouring} and its variants, the natural parameter is the number of available colours~$k$.
Few parameterized results for {\sc Colouring} restricted to $H$-free graphs are known.  Below we survey 
some initial results.

\begin{theorem}\label{t-hfreefpt}
Let $H$ be a graph. Then the following hold:
\begin{itemize}
\item [(i)] {\sc Colouring} is para-\NP-complete for $H$-free graphs when parameterized by 
$k$ if  $H$ is not a linear forest or if $H$ contains an induced
subgraph isomorphic to $P_6$.
\item [(ii)]  {\sc Colouring} is polynomial-time solvable on $H$-free graphs if $H$ is an induced subgraph of  $P_1+P_3$ or of~$P_4$.
\item [(iii)] {\sc List Colouring} is \FPT\ for $H$-free graphs when parameterized by $k+r$ 
 if $H=rP_1+P_2$.
\item[(iv)]  {\sc List Colouring}  is \FPT\ for $H$-free graphs when parameterized by $k$ if $H$ is an induced subgraph of  $P_1+P_3$ or of~$P_4$.
\end{itemize}
\end{theorem}

\begin{proof}
The first part follows from Theorem~\ref{t-polysum}~(i).  The second part is a restatement of  Theorem~\ref{t-kktw}~(i) (stated again to provide a complete statement on parameterized complexity).  The third part is a result of Couturier et al.~\cite{CGKP12} (they also showed that {\sc Colouring} restricted to $(rP_1+P_2)$-free graphs admits a polynomial kernel 
(see~\cite{Niedermeier06} for a definition)
for every $r\geq 2$, when parameterized
by~$k$).
Couturier et al.~\cite{CGKP12} also proved the {\sc List Colouring} result for $(P_1+P_3)$-free graphs; the result for $P_4$-free graphs was shown by Jansen and Scheffler~\cite{JS97},
who described a linear time algorithm.
\qed
\end{proof}

These results tell us (also see~\cite{CGKP12}) the smallest open cases:

\begin{open}\label{o-2p2}
Is {\sc Colouring} \FPT\ for $2P_2$-free graphs or for $(2P_1+P_3$)-free graphs when parameterized by $k$?
\end{open}

The same question can also be asked for {\sc Precolouring Extension}. 
In fact, we can also see from Theorem~\ref{t-hfreefpt} that the cases $H=2P_2$ and $H=2P_1+P_3$ are 
the two smallest open cases when we consider {\sc List Colouring} for $H$-free graphs parameterized by the number of colours.
Another natural parameter for {\sc List Colouring} is the list size. However, Theorem~\ref{t-kktw}~(iii) shows that 
in that case {\sc List Colouring} is  para-\NP-complete for $H$-free graphs whenever $H$ is not isomorphic to~$P_3$ (and polynomial-time solvable
otherwise).

Ho\`ang et al.~\cite{HKLSS10} asked whether {\sc Colouring} is  \FPT\  for $P_5$-free graphs when parameterized by~$k$.
In the light of Open Problem~\ref{o-2p2}, we slightly reformulate their open problem.

\begin{open}
Is {\sc Colouring}, when parameterized by $k$, $\W[1]$-hard  for $P_5$-free graphs?
\end{open}

Another interesting problem is to determine whether {\sc 3-Colouring} is \W[1]-hard for $P_r$-free graphs when parameterized by $r$
(as we have noted though, we currently do not know whether there exists an integer $r$ such that {\sc 3-Colouring} is \NP-complete for $P_r$-free graphs).

\subsubsection{Certifying Algorithms}

Just as with \NP-hard problems it is natural to try to refine our understanding by asking about fixed-parameter tractability, for problems in  \cP, we ask for polynomial-time  algorithms  that not only find solutions but also  
provide  certificates which demonstrate the correctness of solutions and can be ``easily'' verified.
These algorithms are called {\it certifying} (see, for example, the survey of 
McConnell, Mehlhorn, N\"aher and Schweitzer~\cite{MMNS11}).

For {\sc Colouring}, if the input graph $G=(V,E)$ does have the sought $k$-colouring, then a certifying algorithm can give the colouring as a certificate.    
If $G$ does not have a $k$-colouring, then it must have an induced subgraph that is $(k+1)$-vertex-critical (just delete vertices until one is reached).  If for some class of graphs that is closed under vertex deletion, it is possible to construct the set of all the $(k+1)$-vertex-critical graphs (and this set is finite), then a certifying algorithm for {\sc $k$-Colouring} 
 for that graph class can, when the input graph $G$ is not $k$-colourable, give as a certificate a graph.  To verify the certificate, one must check that
it is an induced subgraph of $G$ and that it is one of the $(k+1)$-vertex-critical graphs for the class.  

We say that a graph $G$ is \emph{$(k+1)$-critical} with respect to a graph class ${\cal G}$ if $\chi(G)=k+1$ and every proper subgraph of $G$ that belongs to ${\cal G}$ is $k$-colourable. 
We will not go through the details, but clearly one can take a similar approach as above using $(k+1)$-critical graphs (rather than $(k+1)$-vertex-critical graphs).
We note that Ho\`ang, Moore, Recoskie, Sawada and Vatshelle~\cite{HMRSV15} observed that if a graph class
has a finite number of $(k+1)$-critical graphs, then
it has a finite number of $(k+1)$-vertex-critical graphs. 

Due to Theorem~\ref{t-polysum}~(iii):4, the case $H=P_5$ is a natural starting point. 
Two certifying algorithms exist for 
$3$-{\sc Colouring} on $P_5$-free graphs. The first one is due to
Bruce, Ho\`ang, and Sawada~\cite{BHS09}. They showed that there exist six 4-critical $P_5$-free graphs in total and gave an explicit construction of these graphs.
The same authors asked whether there exists an algorithm faster than brute force for checking whether a graph contains one of these six 4-critical $P_5$-free graphs as a subgraph.
The second certifying algorithm is due to Maffray and Morel~\cite{MM12}. They showed that there exist twelve 4-vertex-critical $P_5$-free graphs in total and gave an explicit construction of these graphs. The running time of the corresponding  certifying algorithm of Maffray and Morel~\cite{MM12} is linear (and as such answered the question posed by
Bruce et al.~\cite{BHS09}).

For all $k\geq 5$, Ho\`ang, Moore, Recoskie, Sawada and Vatshelle~\cite{HMRSV15} constructed  an infinite set of $k$-vertex-critical $P_5$-free graphs which, as noted, implies that  the set of $k$-critical $P_5$-free graphs is also infinite.
For the case $k=5$, they used an exhaustive computer search to construct an infinite set of $k$-critical $P_5$-free graphs. 

Chudnovsky, Goedgebeur, Schaudt and Zhong~\cite{CGSZ15} proved that there exist 24 4-critical $P_6$-free graphs and 80 4-vertex-critical $P_6$-free graphs. Hence, their result implies the existence of a certifying algorithm that solves $3$-{\sc Colouring} for $P_6$-free graphs (which answers an open problem 
of a previous version of this survey). The same authors also proved that there are infinitely many $4$-critical $P_7$-free graphs. 
Moreover, they observed that there are infinitely many $4$-critical $H$-free graphs if $H$ contains a cycle or a claw (which was to be expected given 
Theorem~\ref{t-polysum}~(i):1-2).
Hence they  showed that
for a connected graph $H$, there are finitely many $4$-critical $H$-free graphs if and only if $H$ is a subgraph of $P_6$. 

\begin{open}\label{o-certifying}
Determine all linear forests $H$, for which there exists a certifying algorithm that solves $3$-{\sc Colouring} for $H$-free graphs.
\end{open}

\subsection{Choosability}\label{ss-choos}
	
Golovach and Heggernes~\cite{GH09} showed that {\sc Choosability} is \NP-hard for $P_5$-free graphs.
Their work was continued by Golovach, Heggernes, van 't Hof and Paulusma who 
implicitly showed the following result in the proof 
of~\cite[Theorem 2]{GHHP12}\footnote{See the Appendix for an explicit proof of this result.}
({\it adding} a dominating vertex to a graph means creating a new vertex and making it adjacent to every existing vertex).

\begin{theorem}\label{t-dominating}
Let ${\cal G}$ be a graph class that is closed under adding dominating vertices.
If {\sc Colouring} is \NP-hard for ${\cal G}$, then {\sc Choosability} is \NP-hard for ${\cal G}$.
\end{theorem}

Golovach et al.~\cite{GHHP12} then used Theorem~\ref{t-dominating} to prove the following result.

\begin{theorem}\label{t-ghhp12}
Let $H$ be a graph. Then the following hold:
\begin{itemize}
\item [(i)] If $H\notin \{K_{1,3},P_1,2P_1,3P_1,P_1+P_2,P_1+P_3,P_2,P_3,P_4\}$ then
{\sc Choosability} is \NP-hard for $H$-free graphs. 
\item [(ii)] If $H\in \{P_1,2P_1,3P_1,P_2,P_3\}$ then {\sc Choosability} is
polynomial-time solvable for $H$-free graphs.
\end{itemize}
\end{theorem}
Note that there are four missing cases in Theorem~\ref{t-ghhp12}: 
when $H\in \{K_{1,3},P_1+P_2,P_1+P_3,P_4\}$. 

The following result is due to Gutner~\cite{Gu96}.

\begin{theorem}
\label{t-gutner}
{\sc $3$-Choosability} and {\sc $4$-Choosability} are $\Uppi_2^p$-complete for planar graphs.
\end{theorem}

Gutner and Tarsi~\cite{GT09} showed the following result.

\begin{theorem}
\label{t-gutnertarsi}
For all $k\geq 3$,  {\sc $k$-Choosability} is $\Uppi_2^p$-complete on bipartite graphs.
\end{theorem}

Hence, for some graphs $H$, Theorem~\ref{t-ghhp12} can be strengthened: statements of NP-hardness can be replaced by stronger statements of $\Uppi^p_2$-hardness.

\begin{theorem}\label{t-uppi}
Let $H$ be a graph. Then {\sc Choosability} is $\Uppi_2^p$-hard for $H$-free graphs if $H$ is non-planar or contains an odd cycle.
\end{theorem}

We describe two open problems for {\sc Choosability}.  The first asks for the resolution of the missing cases of Theorem~\ref{t-ghhp12}.

\begin{open}\label{o-first}
Is {\sc Choosability} \NP-hard for $H$-free graphs if $H\in \{K_{1,3},P_1+P_2,P_1+P_3,P_4\}$?
\end{open}

We observe that for $H\in \{P_1+P_2,P_1+P_3,P_4\}$,  the class of $H$-free graphs contains the class of complete bipartite graphs as a subclass. As noted 
by Golovach et al.~\cite{GHHP12}, the computational complexity of {\sc Choosability} on complete bipartite graphs is still open as well. We discuss the fourth open case $H=K_{1,3}$ in more detail later.

The second open problem for {\sc Choosability}
asks for an extension of Theorem~\ref{t-uppi}: 

\begin{open}\label{o-second}
Is {\sc Choosability} $\Uppi_2^p$-hard for \emph{all} those classes for which it is \NP-hard.
\end{open}

If $H\in \{P_1+P_2,P_1+P_3,P_4\}$, then the class of $H$-free graphs contains the class of complete bipartite graphs as a subclass. 
Even the complexity status of {\sc Choosability} for complete bipartite graphs is open. This could be a possible direction for further research.
We also make the following remark, which shows that another natural approach does \emph{not} work.
In contrast to {\sc Precolouring Extension}, there exist graphs $H$ for which {\sc List Colouring} is \NP-complete when restricted to $H$-free graphs, while {\sc Colouring} becomes polynomial-time solvable. 
However, it is not possible (unfortunately) to strengthen Theorem~\ref{t-dominating} by replacing the \NP-hardness of {\sc Colouring}
by \NP-hardness of 
{\sc List Colouring} as a sufficient condition for \NP-hardness of {\sc Choosability}. For instance, let ${\cal G}$ be the class of $(3P_1,P_1+P_2)$-free
graphs.
It is known that {\sc List Colouring} is \NP-complete for this graph class~\cite{GPS12}, which is closed under adding of dominating vertices, while {\sc Choosability} is polynomial-time solvable even for $3P_1$-free graphs due to Theorem~\ref{t-ghhp12}.

As an aside, there also exist graph classes for which {\sc Precolouring Extension} is \NP-hard but  {\sc Choosability} is polynomial-time solvable.
Galvin~\cite{Ga95} showed that every line graph of a bipartite graph is $k$-choosable if and only if it is $k$-colourable.
Because line graphs of bipartite graphs are  
perfect~\cite{Ga58}, and {\sc Colouring} can be solved in polynomial time on perfect graphs~\cite{GLS84}, this means that {\sc Choosability} is polynomial-time solvable on such graphs. However, {\sc Precolouring Extension} is \NP-complete even for
line graphs of complete bipartite graphs, as shown by Hujter and Tuza~\cite{HT93}.

\medskip
\noindent
We now consider $k$-{\sc Choosability}.
Golovach and Heggernes~\cite{GH09} showed that $k$-{\sc Choosability} is linear-time solvable on $P_5$-free graphs.
Golovach et al.~\cite{GHHP12} extended this result and proved statement~(i) of Theorem~\ref{t-fpt} below. Statement~(ii) of this theorem 
follows from Theorem~\ref{t-gutnertarsi}, whereas statement (iii) follows from Theorem~\ref{t-gutner}.
Also recall that $2$-{\sc Choosability} is polynomial-time solvable for general graphs by Theorem~\ref{t-general}.

\begin{theorem}\label{t-fpt}
Let $H$ be graph. Then the following three statements hold for $H$-free graphs:
\begin{itemize}
\item [(i)] For all $k\geq 1$, $k$-{\sc Choosability} is linear-time solvable if $H$ is a linear forest.
\item [(ii)] For all $k\geq 3$, $k$-{\sc Choosability} is $\Uppi_2^p$-hard if $H$ contains an odd cycle.
\item [(iii)] For $3\leq k\leq 4$, $k$-{\sc Choosability} is $\Uppi_2^p$-hard if $H$ is non-planar.
\end{itemize}
\end{theorem}
Theorem~\ref{t-fpt} leads to the following open problem.

\begin{open}\label{o-k13}
For all $k\geq 3$, determine the complexity of $k$-{\sc Choosability} on $H$-free graphs when $H$ is a 
bipartite graph that is not a linear forest.
\end{open}

Open Problem~\ref{o-k13}  
seems difficult
due to its connection to the 
well-known and long-standing List Colouring Conjecture, for which the aforementioned  result of Galvin~\cite{Ga95} is a special case.
This conjecture states that every line graph is $k$-choosable if and only if it is $k$-colourable. This conjecture is usually attributed to Vizing (cf.~\cite{HC92}).
As observed by Golovach et al.~\cite{GHHP12}, $k$-{\sc Choosability} is \NP-hard on $K_{1,3}$-free graphs for every $k\geq 3$ if the List Colouring Conjecture is true. This could mean that Theorem~\ref{t-fpt}~(i) is best possible.

\section{Results and Open Problems for $(H_1,H_2)$-Free Graphs}\label{s-hfreetwo}

When we forbid two induced subgraphs, only partial results are known for {\sc Colouring} and its variants. 
We survey these results below. First we need some other other results starting with the following theorem of Maffray and Preissmann~\cite{MP96}.

\begin{theorem}
\label{t-mp96}
$3$-{\sc Colouring} is \NP-complete for $C_3$-free graphs of maximum degree at most~$4$.
\end{theorem}

For $1\leq h\leq i\leq j$, let $S_{h,i,j}$ denote the tree that is the union of paths of lengths $h$, $i$ and $j$ whose only common vertex is an end-vertex of each.
Observe that $S_{1,1,1}=K_{1,3}$,
$S_{1,1,2}$ is the 
chair and $S_{1,2,2}$ is the ``E''-graph (see Figure~\ref{f-small}). 
Let $A_{h,i,j}$ denote the line graph of $S_{h,i,j}$.
Schindl~\cite{Sc05} showed the following result. 

\begin{theorem}
\label{t-schindl}
Let $\{H_1,\ldots,H_p\}$ be a finite set of graphs.
Then {\sc Colouring} is \NP-complete for $(H_1,\ldots,H_p)$-free graphs if the complement of each $H_i$ has a connected component that is  isomorphic neither 
to any graph $A_{h,i,j}$, for $1\leq h\leq i\leq j$, nor to any path 
$P_r$ for $r\geq 1$.
\end{theorem}

We also need the following result due to Gravier, Ho\`ang and Maffray~\cite{GHM03} (which is a slight improvement on a similar result of
Gy\'arf\'as~\cite{Gy87}).

\begin{theorem}
\label{t-gy}
Let $r,t\geq 1$ be two integers.
Then every $(K_r,P_t)$-free graph can be coloured with at most $(t-2)^{r-2}$ colours.
\end{theorem}

We note that Theorem~\ref{t-gy} has been improved by 
Esperet, Lemoine, Maffray and Morel~\cite{ELMM13} for the case $r=4, t=5$; they 
showed that  every $(K_4,P_5)$-free graph is 5-colourable.

It can be seen that {\sc Colouring} is polynomial-time solvable on any graph class of bounded clique-width by combining two results:
Kobler and Rotics~\cite{KR03} showed that for any constant $q$, {\sc Colouring}  is polynomial-time solvable if a $q$-expression
is given 
(they also showed that in that case {\sc List $k$-Colouring}
is linear-time solvable
for all $k\geq 1$), 
and Oum~\cite{Oum08} showed that a  $(8^{p}-1)$-expression for any $n$-vertex graph with clique-width at most $p$ can be found in $O(n^3)$ time.

\begin{theorem}\label{t-cliquewidth}
Let ${\cal G}$ be a graph class of bounded clique-width. The following two statements hold:
\begin{itemize}
\item [(i)]
{\sc Colouring} can be solved in polynomial time on ${\cal G}$.
\item [(ii)] For all  $k\geq 1$, {\sc List $k$-Colouring}  can be solved in polynomial time on ${\cal G}$.
\end{itemize}
\end{theorem}

As an aside, the statement of Theorem~\ref{t-cliquewidth}~(i) is valid neither for  {\sc Precolouring Extension} nor for {\sc List Colouring}.
For instance, Bonomo, Dur\'{a}n and Marenco~\cite{BDM09} proved that {\sc Precolouring Extension} is \NP-complete for distance-hereditary graphs, which have clique-width at most~3~\cite{GR00}, whereas,
by Theorem~\ref{t-cb},  
even 3-{\sc List Colouring} is \NP-complete for complete bipartite graphs, which 
have clique-width at most~2~\cite{CO00}.

The graph $\overline{P_1+P_3}$ is called the \emph{paw}
(see Figure~\ref{f-small});
we also denote it by $C_3^+$.
By using a result of Olariu~\cite{Ol88}, which states that a graph is $C_3^+$-free if and only if it is $C_3$-free or a complete 
multipartite graph, Kr\'al' et al.~\cite{KKTW01} observed the following.

\begin{theorem}
\label{t-paw}
Let $H$ be a graph. 
Then {\sc Colouring} is polynomial-time solvable on $(C_3,H)$-free graphs if and only if it is polynomial-time solvable for $(C_3^+,H)$-free graphs.
\end{theorem}
 
\begin{figure}[ht]
\centering\scalebox{0.7}{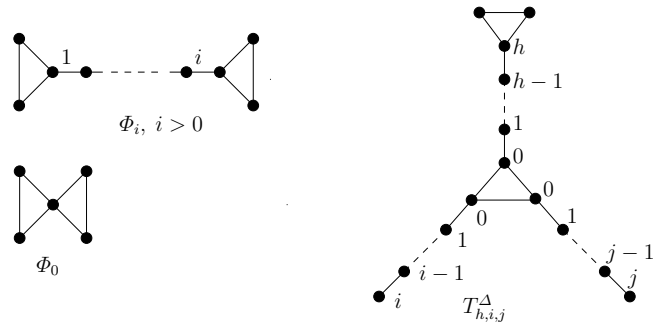}
\caption{The graphs $\Phi_i$ and $T_{h,i,j}^{\Delta}$. 
\label{f-t}}
\end{figure}

Theorem~\ref{t-twographs} below summarizes  results on {\sc Colouring} for graph classes defined by two forbidden induced subgraphs.
In order to state this theorem, we need to define the following graphs. 
The graph $\overline{2P_1+P_2}$ is also called a {\it diamond}.
The graph $\overline{P_1+P_4}$ is also called the \emph{gem}.
The graph $\overline{P_5}$ is also called the \emph{house}.
The graph $\overline{C_4+P_1}$ is also called the \emph{butterfly}.
These graphs are all shown in Figure~\ref{f-small} as are the \emph{hammer} and the \emph{bull} which we also denote by $C_3^*$ and $C_3^{++}$ respectively
(recall that $C_3^+=\overline{P_1+P_3}$ denotes the paw). 
The graph $\Phi_i$, $i \geq 0$, is composed of a path $P$ on $i$ edges with end-vertices $u$ and $v$ and a $K_3$ that intersects $P$ in $u$ and a $K_3$ that intersects $P$ in $v$
(notice that if $i=0$, then $u=v$ so $\Phi_0$ is the butterfly). 
The graph $T_{h,i,j}^\Delta$, $h,i,j \geq 0$, is composed of a $\Phi_h$, a $\Phi_i$ and a $\Phi_j$ which all intersect in a $K_3$ in such a way that each of its vertices has degree at most~3.  Both graphs are illustrated in Figure~\ref{f-t}.  The graph $\Phi_{i,j}$, $i, j \geq 0$, is composed of a path $P_1$ on $i$ edges with end-vertices $u$ and $v$, a path $P_2$ on $j$ edges with end-vertices $w$ and $x$ and a $K_3$ that intersects $P_1$ in $u$, a $K_3$ that intersects $P_1$ in $v$ and $P_2$ in $w$, $v \neq w$, and a $K_3$ that intersects $P_2$ in $x$ (notice that possibly $u=v$ or $w=x$).
The graph $\Phi'_i$, $i \geq 0$, is formed from $\Phi_i$ by adding a new vertex of degree 1 adjacent to one of the two $K_3$s but not to the path between them.  The graph $\Phi''_i$, $i \geq 0$, is formed from $\Phi_i$ by adding two new vertices of degree 1; each one is adjacent to a different one of the two $K_3$s but not to the path between them.  These graphs are illustrated in Figure~\ref{f-fififi}.

\begin{figure}[ht]
\centering\scalebox{0.7}{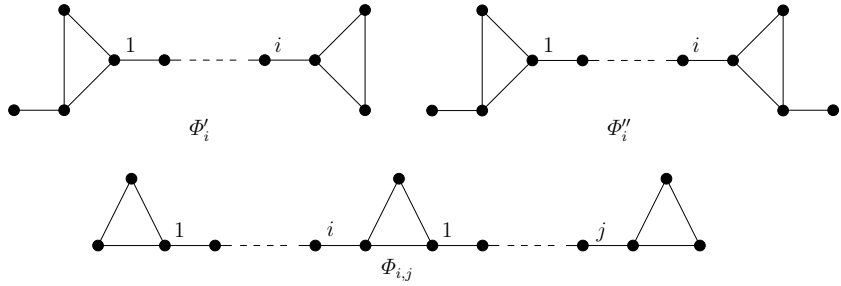}
\caption{The graphs $\Phi_i'$, $\Phi_i''$ and $\Phi_{i,j}$. 
\label{f-fififi}}
\end{figure}

A (partial) proof of Theorem~\ref{t-twographs} can be found in the papers of 
Golovach and Paulusma~\cite{GP} and Dabrowski, Golovach and Paulusma~\cite{DGP}. 
Note that, by symmetry, the graphs $H_1$ and $H_2$ may be swapped in each of the subcases of Theorem~\ref{t-twographs}.

\begin{theorem}\label{t-twographs}
Let $H_1$ and $H_2$ be two graphs. Then the following hold:
\begin{itemize}
\item [(i)] {\sc Colouring} is \NP-complete for $(H_1, H_2)$-free graphs if\\[-8pt]
\begin{enumerate}
\item  $H_1\si C_r$ for $r\geq 3$, and $H_2\si C_s$ for  $s\geq 3$
\item $H_1\si K_{1,3}$, and $H_2\si K_{1,3}$ or $H_2\si
\overline{2P_1+P_2}$ or $H_2\si C_r$ for  $r\geq 4$ or $H_2\si K_4$ or
$H_2\si \Phi_{i,j}$ for $i,j\geq 0$, both even or $H_2\si
\Phi'_i$ for $i\geq 1$ odd or $H_2\si \Phi''_i$ for $i\geq 0$ even
\item $H_1\si \overline{\Phi_i}$ for $i\geq 1$, and $H_2$ contains a spanning subgraph of $2P_2$ as an induced subgraph
\item  $H_1$ and $H_2$ contain a spanning subgraph of $2P_2$ as an induced subgraph
\item $H_1\si C_3^{++}$, and $H_2\si K_{1,4}$ or $H_2\si \overline{C_4+P_1}$
\item  $H_1\si C_3$, and $H_2\si K_{1,r}$  for  $r\geq 5$
\item  $H_1\si C_3$, and 
$H_2\si P_{22}$
\item $H_1\si C_r$ for  $r\geq 5$, and $H_2$ contains a spanning subgraph of $2P_2$ as an induced subgraph
\item $H_1\si C_r+ P_1$ for $3\leq r\leq 4$ or $H_1\si  \overline{C_r}$ for $r\geq 6$, and
$H_2$ contains a spanning subgraph of $2P_2$ as an induced subgraph
\item $H_1\si K_5$, and $H_2\si P_7$
\item $H_1\si K_6$, and $H_2\si P_6$.\\[-8pt]
\end{enumerate}
\item [(ii)] {\sc Colouring} is polynomial-time solvable for $(H_1, H_2)$-free graphs if\\[-8pt] 
\begin{enumerate}
\item $H_1$ or $H_2$ is an induced subgraph of $P_1+P_3$ or of $P_4$
\item $H_1\ssi K_{1,3}$, and 
$H_2\ssi C_3^{++}$ or $H_2\ssi C_3^*$ or $H_2\ssi P_5$
\item $H_1\neq K_{1,5}$ is a forest on at most six vertices or $H_1= K_{1,3}+3P_1$, and $H_2\ssi C_3^+$  
\item $H_1\ssi sP_2$ or $H_1\ssi sP_1+P_5$ for  $s\geq 1$, and $H_2=K_t$ for $t\geq 4$
\item $H_1\ssi sP_2$ or $H_1\ssi sP_1+P_5$ for $s\geq 1$, and $H_2\ssi C_3^+$
\item $H_1\ssi P_1+P_4$ or $H_1\ssi P_5$, and $H_2\ssi  \overline{P_1+P_4}$ 
\item $H_1\ssi P_1+P_4$ or 
$H_1\ssi P_5$,
and $H_2\ssi \overline{P_5}$
\item $H_1\ssi 2P_1+P_2$, and 
$H_2\ssi \overline{P_1+2P_2}$ or
$H_2\ssi \overline{2P_1+P_3}$ or $H_2\ssi \overline{P_2+P_3}$
\item  $H_1\ssi \overline{2P_1+P_2}$, and
$H_2\ssi P_1+2P_2$ or
 $H_2\ssi {2P_1+P_3}$ or $H_2\ssi P_2+P_3$
\item $H_1\ssi sP_1+P_2$ for $s\geq 0$ or
$H_1=P_5$, and  $H_2\ssi \overline{tP_1+P_2}$ for $t\geq 0$
\item $H_1\ssi 4P_1$ and $H_2 \ssi \overline{2P_1+P_3}$
\item $H_1\ssi P_5$, and $H_2\ssi C_4$ 
or $H_2\ssi \overline{2P_1+P_3}$.
\end{enumerate}
\end{itemize}
\end{theorem}

\begin{proof}
In each case we either refer back to an earlier result, or give a reference.  The results quoted can clearly be seen to imply the statements of the theorem.
\begin{itemize}
\item [(i)] We first consider the \NP-completeness results.  
\begin{enumerate}
\item By Theorem~\ref{t-kl07}, for $k\geq 3$, $k$-{\sc Colouring}
is \NP-complete for $(C_r,C_s)$-free graphs for all $r\geq 3$ and $s\geq 3$.
\item These cases were proved for {\sc 3-Colouring} by Lozin and Purcell~\cite{LP14}.
 (This is sufficient but we note also that by Theorem~\ref{t-line},  for  $k\geq 3$, $k$-{\sc Colouring} is \NP-complete for claw-free graphs and that  Kr\'al' et al.~\cite{KKTW01} showed that $3$-{\sc Colouring} is \NP-complete for $(C_r,K_{1,3})$-free graphs whenever $r\geq 4$ and for 
 $(K_4,K_{1,3},\overline{2P_1+P_2})$-free  graphs.)
\item This follows from a result of Lozin and Malyshev~\cite{LM15} who proved that {\sc Colouring} is \NP-complete for $(C_3+P_1,2P_2,2P_1+P_2,4P_1,C_5,\overline{C_6},\ldots,\overline{C_p},\overline{\Phi_0},\ldots,\overline{\Phi_p})$-free graphs for every integer $p\geq 6$.
\item This is a result of Kr\'al' et al.~\cite{KKTW01}.
\item  Malyshev~\cite{Ma15} proved that $3$-{\sc Colouring} is \NP-complete for $(C_3^{++},\overline{C_4+P_1},K_{1,4})$-free graphs after previously proving that 
$3$-{\sc Colouring} is \NP-complete for $(C_3^{++},K_{1,4})$-free graphs~\cite{Ma14}. 
Note that Theorem~\ref{t-twographs}~(i).3 already implies that 3-{\sc Colouring} is \NP-complete for $(\overline{C_4+P_1},K_{1,4})$-free graphs.
\item  By Theorem~\ref{t-mp96}, $3$-{\sc Colouring} is \NP-complete for $(C_3,K_{1,r})$-free graphs for all $r\geq 5$.
\item Huang, Johnson and Paulusma~\cite{HJP14} proved that $4$-{\sc Colouring} is \NP-complete for $(C_3,P_{22})$-free graphs, thereby improving
a result of Golovach et al.~\cite{GPS11} who showed that $4$-{\sc Colouring} is \NP-complete for $(C_3,P_{164})$-free graphs.
\item This is a result of
 Kr\'al' et al.~\cite{KKTW01} 
 \item This follows from Theorem~\ref{t-schindl}.
 \item This follows from Theorem~\ref{t-polysum}~(i) and the fact that $K_6$ is not 5-colourable.
 \item This follows from Theorem~\ref{t-polysum}~(i) and the fact that $K_5$ is not 4-colourable.
 \end{enumerate}
\item [(ii)] We now consider the tractable cases. 
\begin{enumerate}
\item This follows from Theorem~\ref{t-kktw}~(i).
\item This was proved by Malyshev~\cite{Ma14} for $(K_{1,3},C_3^*)$-free graphs and $(K_{1,3},P_5)$-free graphs and by Malyshev~\cite{Ma} for $(K_{1,3},C_3^{++})$-free graphs.
\item 
First we consider the case when $H_1$ is a forest 
on at most six vertices 
not isomorphic to $K_{1,5}$ and $H_2\ssi C_3$.
Dabrowski, Lozin, Raman and Ries~\cite{DLRR12} proved that {\sc Colouring} is polynomial-time solvable for $(H_1,C_3)$-free graphs
by combining a number of new results  with known results for $H_1=K_{1,4}$~\cite{KKTW01}, 
$H_1=S_{1,2,2}$~\cite{Ra04}, $H_1=P_2+P_4$~\cite{BGPS12a},
$H_1=2P_3$~\cite{BGPS12b}, $H_1=P_6$~\cite{BKM06}, $H_1$ is the cross~\cite{RS02} and $H_1$ is the ``H''-graph~\cite{Ra04} 
(see Figure~\ref{f-small} for pictures of the cross and the ``H''-graph).
Then they applied Theorem~\ref{t-paw}.
Dabrowski and Paulusma~\cite{DP} proved that the class of 
$(K_{1,3}+3P_1,C_3^+)$-free graphs has bounded clique-width, so
Theorem~\ref{t-cliquewidth}~(i) can be applied.
\item  Theorem~\ref{t-gy} implies that for all $r\geq 1$, {\sc Colouring} is polynomial-time solvable on $(K_r,F)$-free graphs for some linear forest $F$ if $k$-{\sc Colouring} is polynomial-time solvable on $F$-free graphs for all $k\geq 1$. The latter is true for $F=sP_1+P_5$ and $F=sP_2$, for all 
$s\geq 1$, 
by Theorem~\ref{t-polysum}~(iii).
\item 
This is obtained by combining the arguments of the previous case with Theorem~\ref{t-paw}.
\item  The classes of $(P_1+P_4,\overline{P_1+P_4})$-free graphs~\cite{BLM04} and  $(P_5,\overline{P_1+P_4})$-free graphs~\cite{BK05}  have bounded clique-width. Hence, {\sc Colouring} is polynomial-time solvable for these 
two 
graph classes by Theorem~\ref{t-cliquewidth}~(i).
\item For the class of $(P_1+P_4,\overline{P_5})$-free graphs, we again note they  have bounded clique-width~\cite{BK05}. 
Ho\`ang and Lazzarato~\cite{HL} showed that {\sc Colouring} is polynomial-time solvable on $(P_5,\overline{P_5})$-free graphs
(in fact they show that the weighted variant of {\sc Colouring} is polynomial-time solvable).
Previously, this result was known only for $(2P_2,\overline{P_5})$-free graphs, as Ho\`ang, Maffray and Mechebbek~\cite{HMM12} showed that these graphs are b-perfect, and in the same paper they proved that {\sc Colouring} is polynomial-time solvable for 
b-perfect graphs.
\item 
Dabrowski, Dross and Paulusma~\cite{DDP} proved that the class of $(2P_1+P_2,\overline{P_1+2P_2})$-free graphs has bounded clique-width.
Dabrowski, Huang and Paulusma~\cite{DHP} showed that the class of $(2P_1+P_2,\overline{2P_1+P_3})$-free graphs and the class of $(2P_1+P_2,\overline{P_2+P_3})$-free graphs have bounded clique-width.
\item This is due to Dabrowski et al.~\cite{DDP,DHP} as well. 
\item Dabrowski, Golovach and Paulusma~\cite{DGP} proved that
for every two integers $s\geq 0$ and $t\geq 0$,
 {\sc Colouring} is polynomial-time solvable
for $(sP_1+P_2,\overline{tP_1+P_2})$-free graphs. 
Malyshev and Lobanova~\cite{MO} proved that, for all $t\geq 0$, {\sc Colouring} is polynomial-time solvable for
$(P_5, \overline{tP_1+P_2})$-free graphs, which generalizes an earlier result of Dabrowski, Golovach and Paulusma~\cite{DGP}
for the class of $(2P_2, \overline{tP_1+P_2})$-free graphs. 
\item The class of $(4P_1,\overline{2P_1+P_3})$-free graphs has bounded clique-width~\cite{DHPb}, hence we apply Theorem~\ref{t-cliquewidth}~(i). 
\item This was proved by Malyshev~\cite{Ma14} 
for  $(P_5,C_4)$ and by Malyshev~\cite{Ma} for $(P_5,\overline{2P_1+P_3})$. 
Previously this was known for $(P_5,\overline{2P_1+P_2})$-free graphs~\cite{AM02}.
\qed
\end{enumerate}
\end{itemize} 
\end{proof}

We pose the following problem.

\begin{open}
Complete the classification of the complexity of {\sc Colouring}  for $(H_1,H_2)$-free graphs.
\end{open}

\noindent
{\bf Some Important Subproblems}
A classification of the complexity of {\sc Colouring} for $(H_1,H_2)$-free graphs is already problematic when $H_1$ and $H_2$ are small.
Lozin and Malyshev~\cite{LM15} determined the computational complexity of {\sc Colouring} restricted to ${\cal H}$-free graphs for every finite set ${\cal H}$ that consists only of 4-vertex graphs except in the following four cases, which are still open: 
\begin{itemize}
\item [$\bullet$] ${\cal H}=\{K_{1,3},4P_1\}$
\item   [$\bullet$]  ${\cal H}=\{K_{1,3},2P_1+P_2,4P_1\}$
\item  [$\bullet$]  ${\cal H}=\{K_{1,3},2P_1+P_2\}$ 
\item  [$\bullet$]  ${\cal H}=\{C_4,4P_1\}$. 
\end{itemize}
The same authors showed that the cases
${\cal H}=\{K_{1,3},2P_1+P_2\}$ and  ${\cal H}=\{K_{1,3},2P_1+P_2,4P_1\}$
are polynomially equivalent (hence three cases remain).
Fraser, Hamel and Ho\`ang~\cite{FHH} continued this study by showing that  {\sc Colouring} 
is polynomial-time solvable for a subclass of $(K_{1,3},4P_1\})$-free graphs, namely for
(a superclass of) $4P_1$-free line graphs.

The above open cases are part of a larger set of open subproblems. As we will see (Theorem~\ref{t-com2}), 
for all $k,r,s,t\geq 1$, {\sc $k$-Colouring} can be solved in linear time for the class of $(K_{r,s},P_t)$-free graphs. However, for {\sc Colouring} restricted to $(K_{r,s},P_t)$-free graphs  
much less is known. Besides the aforementioned open cases, also the case 
of $(C_4,P_6)$-free graphs is a natural open case to consider, as {\sc Colouring} is polynomial-time solvable for $(C_4,P_5)$-free graphs due to Theorem~\ref{t-twographs} (ii).12.
In addition, the case of $(C_3,P_7)$-free graphs still needs to be solved (we discuss this case in more detail later).

Finally, Dabrowski et al.~\cite{DDP} listed all 13 classes of $(H_1,H_2)$-free graphs, for which {\sc Colouring} could still be solved in polynomial time by showing that their clique-width is bounded.
These classes are
\begin{enumerate}
\item $\overline{H_1}\in \{3P_1,P_1+P_3\}$ and $H_2\in \{P_1+S_{1,1,3},S_{1,2,3}\}$;\\[-10pt]
\item $H_1=2P_1+\nobreak P_2$ and $\overline{H_2} \in \{P_1+\nobreak P_2+\nobreak P_3, \allowbreak 
P_1+\nobreak P_5\}$;\\[-10pt]
\item $H_1=\overline{2P_1+\nobreak P_2}$ and $H_2 \in \{P_1+\nobreak P_2+\nobreak P_3,\allowbreak 
P_1+\nobreak P_5\}$;\\[-10pt]
\item $H_1=P_1+P_4$ and $\overline{H_2} \in \{P_1+2P_2,P_2+P_3\}$;\\[-10pt]
\item $\overline{H_1}=P_1+P_4$ and ${H_2} \in \{P_1+2P_2,P_2+P_3\}$;\\[-10pt]
\item $H_1=\overline{H_2}=2P_1+P_3$.
\end{enumerate}

We now discuss what is known for {\sc Precolouring Extension} restricted to $(H_1,H_2)$-free graphs.
By Theorem~\ref{t-kktw}~(ii),  {\sc Precolouring Extension}  can be solved in polynomial time on $(H_1,H_2)$-free graphs
whenever, for some $i\in \{1,2\}$, $H_i\ssi P_4$ or $H_i\ssi P_1+P_3$,
and, of course, the \NP-completeness results from Theorem~\ref{t-twographs} also hold for  {\sc Precolouring Extension}.
This is all that seems to be known of {\sc Precolouring Extension} on $(H_1,H_2)$-free graphs. 
Let us give an example of a class of $(H_1,H_2)$-free graphs for which the complexities of {\sc Colouring}
and {\sc Precolouring Extension} are different (unless \cP$=$\NP):
Case~(ii):3 of Theorem~\ref{t-twographs}, which shows that {\sc Colouring} is polynomial-time solvable for $(C_3,P_6)$-free graphs, can be compared with the following result.

\begin{theorem}\label{t-example}
{\sc Precolouring Extension} is \NP-complete for $(C_3,P_6)$-free graphs.
\end{theorem}

\begin{proof}
We reduce from the restriction of {\sc List Colouring}  to complete bipartite graphs which is
\NP-complete by Theorem~\ref{t-cb}.
Let $G=(V,E)$ be a complete bipartite graph with list assignment~$L$. Let $X=\bigcup_{u\in V}L(u)$.
For each $u\in V$,  add $|X|-|L(u)|$ new vertices, add an edge from each to~$u$, and assign each a different colour from $X\setminus L(u)$. 
Let  $G'$ be the resulting graph,  let $W$ be the set of vertices in $G'-V$ and let $k=|X|$, and
notice that in the previous sentence we have defined a $k$-precolouring of $G'$ in which a vertex
has a colour if and only if it is in~$W$.
It is readily seen that~$G'$ is $(C_3,P_6)$-free, and that~$G$ has a colouring that respects~$L$ if and only if the $k$-precolouring of~$G'$ can be extended to a $k$-colouring. 
\qed
\end{proof}

Hujter and Tuza asked for which graph classes {\sc Precolouring Extension} is \NP-complete  (Problem 1.1 in~\cite{HT96}). We pose the following problem.

\begin{open}
Complete the classification of the complexity of 
{\sc Precolouring Extension} for $(H_1,H_2)$-free graphs.
\end{open}

By combining a number of known hardness results on {\sc List Colouring} for complete bipartite graphs~\cite{JS97},
complete split graphs~\cite{GH09} and $(3P_1,P_1+P_2$)-free graphs~\cite{GPS12} with a number of new hardness results, 
Golovach and Paulusma~\cite{GP} completely classified the complexity of {\sc List Colouring} and 
$\ell$-{\sc List Colouring}, $\ell\geq 3$, for $(H_1,H_2)$-free graphs.
Note that, by symmetry, the graphs $H_1$ and $H_2$ may be swapped in each of the three subcases of Theorem~\ref{t-main}.

\begin{theorem}
\label{t-main}
Let $H_1$ and $H_2$ be two  graphs. Then {\sc List Colouring} is polynomial-time solvable for $(H_1,H_2)$-free graphs 
in the following cases:
\begin{enumerate}
\item $H_1\ssi P_3$ or $H_2\ssi P_3$
\item $H_1\ssi C_3$ and $H_2\ssi K_{1,3}$
\item $H_1=K_r$ for some $r\geq 3$ and $H_2=sP_1$ for some $s\geq 3$.
\end{enumerate}
In all other cases, even {\sc $3$-List Colouring} is \NP-complete for $(H_1,H_2)$-free graphs.
\end{theorem}

The computational complexity classification of {\sc $k$-Colouring}, {\sc $k$-Precolouring Extension} and {\sc List $k$-Colouring} restricted to $(H_1,H_2)$-free graphs is not complete either. 
Tractability for many cases is obtained from 
Theorem~\ref{t-polysum}~(iii)--(iv).
Moreover, as mentioned in the proof of Theorem~\ref{t-twographs}, Cases~(i):1, 2, 4--6, 10 of Theorem~\ref{t-twographs} hold for 3-{\sc Colouring} and Case~(i):7  holds for 4-{\sc Colouring}.
In particular, the case in which $H_1$ is a cycle and $H_2$ is a path has been studied for all three variants~\cite{GPS11,HH13,HJP14}. 

We survey the known results for these two cases below. In order to do this we need three additional results.  
The first additional result was proven by Golovach et al.~\cite{GPS11}, 
which extends a corresponding result of Lozin and Rautenbach~\cite{LR03} from $r=1$ to arbitrary $r\geq 1$. 

\begin{theorem}
\label{t-com2}
For all $k,r,s,t\geq 1$, {\sc List $k$-Colouring} can be solved in linear time for $(K_{r,s},P_t)$-free graphs.
\end{theorem}

Theorem~\ref{t-com2} implies that for all  $g\geq 5$, $k\geq 1$ and $t\geq 1$, {\sc List $k$-Colouring} can be solved in linear time for $P_t$-free graphs of girth at least $g$, or equivalently $(C_3,\ldots,C_{g-1},P_t)$-free graphs (contrast with Theorem~\ref{t-kl07} on {\sc $k$-Colouring}). Huang et al.~\cite{HJP14} showed that when $C_4=K_{2,2}$ is no longer forbidden the computational complexity
changes again by proving that for all $k\geq 4$ and $g\geq~6$, there exists a constant $t^g_k$ such that $k$-{\sc Colouring} is \NP-complete for
$(C_3,C_5,\ldots,C_{g-1},P_{t^g_k})$-free graphs.

We also need 
another result of Huang et al.~\cite{HJP14}. 
\begin{theorem}
\label{t-hp14}
{\sc List $4$-Colouring} is \NP-complete for $(C_5,C_6,K_4,\overline{P_1+2P_2},\overline{P_1+P_4},P_6)$-free graphs.
\end{theorem}
The third additional result was also proven by Huang et al.~\cite{HJP14}. It strengthens a result of 
Kratochv\'{\i}l~\cite{Kr93} who showed that 5-{\sc Precolouring Extension} is \NP-complete for 
$P_{14}$-free bipartite graphs. 
\begin{theorem}
\label{t-prebipartite}
For all $k\geq 4$, $k$-{\sc Precolouring Extension}  is \NP-complete for $P_{10}$-free bipartite graphs.
\end{theorem}

We are now ready to state Theorem~\ref{t-coloringall}. A proof of this theorem was given by Huang et al.~\cite{HJP14}; as it is obtained by combining a number of results from different papers we include it here as well.

\begin{theorem}\label{t-coloringall}
Let $k,s,t$ be three integers.
The following statements hold for $(C_s,P_t)$-free graphs.
\begin{itemize}
\item[(i)] {\sc List $k$-Colouring} is \NP-complete if\\[-8pt]
\begin{enumerate}
\item  $k\geq 4$,   $s=3$ and $t\geq 8$ 
\item  $k\geq 4$,  $s\geq 5$ and $t\geq 6$.\\[-8pt]
\end{enumerate}
\item[] {\sc List $k$-Colouring} is polynomial-time solvable if\\[-8pt]
\begin{enumerate}
  \setcounter{enumi}{2}
\item   $k\leq 2$, $s\geq 3$ and $t\geq 1$\\[-10pt]
\item   $k=3$, $s=3$ and $t\leq 7$
\item   $k=3$, $s=4$ and $t\geq 1$
\item   $k=3$, $s\geq 5$ and $t\leq 6$\\[-10pt]
\item  $k\geq 4$, $s=3$ and $t\leq 6$
\item  $k\geq 4$, $s=4$ and $t\geq 1$
\item  $k\geq 4$, $s\geq 5$ and $t\leq 5$.
\end{enumerate} 
\end{itemize}
\begin{itemize}
\item[(ii)] {\sc $k$-Precolouring Extension} is \NP-complete if\\[-8pt]
\begin{enumerate}
\item $k=4$,   $s=3$ and $t\geq 10$ 
\item $k=4$,  $s=5$ and $t\geq 7$ 
\item $k=4$, $s=6$ and $t\geq 7$ 
\item $k=4$, $s=7$ and $t\geq 8$
\item $k=4$, $s\geq 8$ and $t\geq 7$\\[-10pt] 
\item $k\geq 5$,   $s=3$ and $t\geq 10$ 
\item $k\geq 5$,  $s\geq 5$ and $t\geq 6$.\\[-8pt]
\end{enumerate}
\item[] {\sc $k$-Precolouring Extension} is polynomial-time solvable if\\[-8pt]
\begin{enumerate}
  \setcounter{enumi}{7}
\item  $k\leq 2$, $s\geq 3$ and $t\geq 1$\\[-10pt]
\item   $k=3$, $s=3$ and $t\leq 7$
\item   $k=3$, $s=4$ and $t\geq 1$
\item   $k=3$, $s\geq 5$ and $t\leq 6$\\[-10pt]
\item  $k\geq 4$, $s=3$ and $t\leq 6$
\item  $k\geq 4$, $s=4$ and $t\geq 1$
\item  $k\geq 4$, $s\geq 5$ and $t\leq 5$.
\end{enumerate}
\end{itemize}
\begin{itemize}
\item[(iii)] $k$-{\sc Colouring} is \NP-complete if\\[-8pt]
\begin{enumerate}
\item  $k=4$, $s=3$ and $t\geq 22$ 
\item  $k=4$,  $s=5$ and $t\geq 7$ 
\item  $k=4$, $s=6$ and $t\geq 7$ 
\item  $k=4$, $s=7$ and $t\geq 9$
\item  $k=4$, $s\geq 8$ and $t\geq 7$\\[-10pt]
\item  $k\geq 5$, $s=3$ and $t\geq t_k$ where $t_k$ is a constant that only depends on $k$
\item  $k\geq 5$,  $s=5$ and $t\geq 7$ 
\item  $k\geq 5$,  $s\geq 6$ and $t\geq 6$.\\[-8pt]
\end{enumerate}
\item[] $k$-{\sc Colouring} is polynomial-time solvable if\\[-8pt]
\begin{enumerate}
  \setcounter{enumi}{8}
\item   $k\leq 2$, $s\geq 3$ and $t\geq 1$\\[-10pt]
\item   $k=3$, $s=3$ and $t\leq 7$
\item   $k=3$, $s=4$ and $t\geq 1$
\item   $k=3$, $s\geq 5$ and $t\leq 7$\\[-10pt]
\item  $k=4$, $s=3$ and $t\leq 6$
\item  $k=4$, $s=4$ and $t\geq 1$
\item  $k=4$, $s=5$ and $t\leq 6$
\item  $k=4$, $s\geq 6$ and $t\leq 5$\\[-10pt]
\item   $k\geq 5$, $s=3$ and $t\leq k+2$
\item   $k\geq 5$, $s=4$ and $t\geq 1$
\item   $k\geq 5$, $s\geq 5$ and $t\leq 5$.
\end{enumerate}
\end{itemize}
\end{theorem}

\begin{proof} Again we either refer back to an earlier result, or give a reference and the results quoted can clearly be seen to imply the statements of the theorem.

We first consider the intractable cases of {\sc List $k$-Colouring}.
For (i).1, we note that Huang et al.~\cite{HJP14} showed that {\sc List 4-Colouring} is \NP-complete for $(C_3,P_8)$-free graphs. 
Theorem~\ref{t-hp14} implies that
{\sc List $4$-Colouring}  is \NP-complete for the class of $(C_5,C_6,P_6)$-free graphs which proves (i).2.

We now consider the tractable cases of {\sc List $k$-Colouring}. 
Theorem~\ref{t-general}~(i) implies (i).3,
whereas (i).4 follows from Theorem~\ref{t-polysum}~(iii).
Theorem~\ref{t-com2} implies (i).5 and (i).8,
whereas (i).6 and (i).9  follow from Theorem~\ref{t-polysum}~(iii). 
Recall that the class of $(C_3,P_6)$-free graphs has bounded clique-width, as shown by Brandst\"adt, Klembt and Mahfud~\cite{BKM06}.
Combining this with Theorem~\ref{t-cliquewidth}~(ii) we find that {\sc List $k$-Colouring} is 
polynomial-time solvable on $(C_3,P_6)$-free graphs for all $k\geq 1$. This proves~(i).7

\medskip
\noindent
We now consider {\sc $k$-Precolouring Extension}.
As the tractable cases all follow from Theorem~\ref{t-coloringall}~(i), 
we are left to consider the \NP-complete cases.
Theorem~\ref{t-prebipartite} implies (ii).1 and (ii).6.
Huang et al.~\cite{HJP14} proved that 4-{\sc Precolouring Extension} is \NP-complete for $(C_7,P_8)$-free graphs, which implies
 (ii).4, and they also proved (ii).7.
We observe that (ii).2, (ii).3 and (ii).5 follow immediately from corresponding results for {\sc $k$-Colouring} as shown by Hell and Huang~\cite{HH13}.

\medskip
\noindent
Finally, we consider $k$-{\sc Colouring}; first the
\NP-complete cases.
Recall that (iii).1 has been shown by Huang et al.~\cite{HJP14},
who also proved (iii).6; they showed that  $t_k\leq k + (k+1)(3\cdot 2^{k-1}-1)$ for all $k\geq 5$.
Golovach et al.~\cite{GPS11} proved that
for all $s\geq 5$, there exists a constant $t(s)$ such that $4$-{\sc Colouring} is
 \NP-complete for $(C_5,\ldots,C_s,P_{t(s)})$-free graphs. In particular, 
they showed that  $4$-{\sc Colouring}  is \NP-complete for $(C_5,P_{23})$-free graphs, and this result has
been strengthened by Hell and Huang~\cite{HH13} who proved all the other \NP-completeness subcases.

We now consider the tractable cases of $k$-{\sc Colouring}. 
Theorem~\ref{t-polysum}~(iii) 
implies (iii).9, (iii).10, (iii).12, (iii).16 and (iii).19. Theorem~\ref{t-com2} implies  (iii).11, (iii).14 and (iii).18.
Case~(ii):3 of Theorem~\ref{t-twographs} implies~(iii).13. Chudnovsky, Maceli, Stacho and Zhong~\cite{CMSZ14} proved (iii).15. Theorem~\ref{t-gy} implies~(iii).17.\qed
\end{proof}
Theorem~\ref{t-coloringall} leaves a number of cases open 
(also see Huang et al.~\cite{HJP14}).
\begin{open}\label{o-cspt}
Determine the complexity of the missing cases from Theorem~\ref{t-coloringall} for $(C_s,P_t)$-free graphs, which are:
\begin{itemize}
\item [(i)] for {\sc List $k$-Colouring} when
\begin{itemize}
\item [$\bullet$] $k=3$, $s=3$ and $t \geq 7$
\item [$\bullet$] $k=3$, $s\geq 5$ and $t\geq 7$
\item[$\bullet$] $k\geq 4$, $s=3$ and $t=7$.\\[-8pt]
\end{itemize}
\item[(ii)] 
for {\sc $k$-Precolouring Extension} when
\begin{itemize}
\item [$\bullet$] $k=3$, $s=3$ and $t \geq 7$
\item [$\bullet$] $k=3$, $s\geq 5$ and $t\geq 7$
\item[$\bullet$] $k=4$, $s=3$ and $7\leq t\leq 9$
\item[$\bullet$] $k=4$, $s\geq 5$ and $t=6$
\item[$\bullet$] $k=4$, $s=7$ and $t=7$
\item [$\bullet$] $k\geq 5$, $s=3$ and $7\leq t\leq 9$\\[-8pt]
\end{itemize}
\item[(iii)]
for {\sc $k$-Colouring} when
\begin{itemize}
\item [$\bullet$] $k=3$, $s=3$ and $t \geq 8$
\item [$\bullet$] $k=3$, $s\geq 5$ and $t\geq 8$
\item[$\bullet$] $k=4$, $s=3$ and 
$7\leq t\leq 21$
\item[$\bullet$] $k=4$, $s\geq 6$ and $t=6$
\item [$\bullet$] $k=4$, $s=7$ and $7\leq  t\leq 8$
\item [$\bullet$] $k\geq 5$, $s=3$ and $k+3\leq t\leq t_k-1$
\item [$\bullet$]  $k\geq 5$, $s=5$ and $t=6$.
\end{itemize}
\end{itemize}
\end{open}

\noindent
{\bf Some Important Subproblems}
Note that as a consequence of Theorem~\ref{t-hp14},  {\sc List $4$-Colouring}  is \NP-complete for $(C_5,P_6)$-free graphs.
As {\sc 4-Colouring}  is polynomial-time solvable for $(C_5,P_6)$-free graphs by Theorem~\ref{t-coloringall}, there exists an integer $k$
and two graphs $H_1$ and
$H_2$ (namely $k=4$, $H_1=C_5$ and $H_2=P_6$) for which the complexity of {\sc $k$-Colouring} and {\sc List $k$-Colouring} is not the same
when restricted to $(H_1,H_2)$-free graphs. Recall that such a situation is not known when we forbid only one induced graph $H$. 
Moreover, when forbidding two graphs $H_1$ and $H_2$ no such complexity jump is known, for any integer $k$, between {\sc $k$-Colouring} and {\sc $k$-Precolouring Extension} restricted to $(H_1,H_2)$-free graphs. 
Hence, it would also
be interesting  to determine the complexity of $4$-{\sc Precolouring Extension} for $(C_5,P_6)$-free graphs,
which is one of the cases in Open Problem~\ref{o-cspt}.

In addition to the above, we point out another case in Open Problem~\ref{o-cspt}: that of determining the complexity of $4$-{\sc Colouring} restricted to $(C_3,P_7)$-free graphs.
By Theorem~\ref{t-gy}, every $(C_3,P_7)$-free graph is 5-colourable. 
By Thereom~\ref{t-coloringall}~(iii), {\sc $k$-Colouring} is polynomial-time solvable for $(C_3,P_7)$-free graphs if $k\leq 3$. Hence, the problems $4$-{\sc Colouring} and {\sc Colouring} are polynomially equivalent for $(C_3,P_7)$-free graphs.

\medskip
\noindent
We now discuss a number of  results for $k$-{\sc Colouring} restricted to $(H_1,H_2)$-free graphs when $(H_1,H_2)$ is not a cycle and a path.

First we consider pairs of graphs $(H_1,H_2)$ with the property that every $(H_1,H_2)$-free graph is 3-colourable. Because 2-{\sc Colouring}
is polynomial-time solvable, such results imply polynomial-time solvability of {\sc 3-Colouring} for $(H_1,H_2)$-free graphs. 

We note that only when $H\in \{P_1,P_2\}$ is every $H$-free graph 3-colourable.  
Thus for all graphs $H_2$, every $(P_1,H_2)$-free graph and every $(P_2,H_2)$-free graph is 3-colourable.   Also Wagon~\cite{Wa80} showed that every $(K_r,2P_2)$-free graph is $\frac{1}{2}r(r-1)$-colourable, which implies that every $(C_3,2P_2)$-free graph is 3-colourable. 

We focus now  on the case where $H_1$ and $H_2$ are connected and show that this is \emph{almost} completely understood.
A pair of graphs $(H_1,H_2)$ is called 
\emph{good} if every $(H_1,H_2)$-free graph is 3-colourable,
and, moreover, the class of $(H_1,H_2)$-free graphs is properly contained in the classes of $H_1$-free graphs and $H_2$-free graphs.

A good pair $(H_1,H_2)$ is {\it saturated} if there is no good pair $(H_1',H_2')$ with $H_1\subsetneq_i H_1'$ and $H_2\subsetneq_i H_2'$.
We note in passing that Sumner~\cite{Su80} showed that 
every $(C_3,P_5)$-free graph is 3-colourable.
However, the pair $(C_3,P_5)$ is not saturated. 
This follows from this result of Randerath~\cite{Ra04} (see Figure~\ref{f-small} for the names of small graphs): 
\begin{itemize}
\item If $(K_3,\mbox{fork})$ is a good pair, then $(K_3,\mbox{fork})$,  $(K_3,\mbox{``H''-graph})$ and $(K_4,P_4)$ are the only saturated pairs of connected graphs.
\item If $(K_3,\mbox{fork})$ is not a good pair, then $(K_3,\mbox{cross})$, $(K_3,\mbox{``E''-graph})$, $(K_3,\mbox{``H''-graph})$ 
and $(K_4,P_4)$ are the only saturated pairs of connected graphs.
\end{itemize}
Note that the cross and ``E''-graph are the two maximal connected
proper induced subgraphs of the fork.
Hence the following open problem remains (which is Conjecture~6 in~\cite{Ra04} and Conjecture~44 in~\cite{RS04b}).

\begin{open}
Is every $(K_3,\mbox{fork})$-free graph $3$-colourable?
\end{open}
Recently, Fan, Xu, Ye and Yu~\cite{FXYY14} made progress in answering this question by proving that every $(C_5,K_3,\mbox{fork})$-free graph  is 3-colourable.

The natural next question is, of course, to ask when $(H_1,H_2)$-free graphs are $k$-colourable for $k \geq 4$.  A little is known. 
For $r\geq 2$ and $s\geq 1$, the {\it broom} $B(r, s)$ (also called a {\it mop}) is the graph obtained from a star $K_{1,s+1}$ after subdividing one of its edges $r-2$ times, so $B(2,s)=K_{1,s+1}$, $B(r,1)=P_{r+1}$ and $B(r,2)=S_{1,1,r-1}$.
Gy\'arf\'as, Szemer\'edi and Tuza~\cite{GST80} proved that every $(C_3,B(r,s))$-free graph is $(r+s-1)$-colourable.
Hence, every $(C_3,P_r)$-free graph is $(r-1)$-colourable.
Randerath and Schiermeyer~\cite{RS04b} improved this implication by showing that for all $r\geq 4$, every $(C_3,P_r)$-free graph is $(r-2)$-colourable. This result has been extended by Wang and Wu~\cite{WW}, who proved that every connected $(C_3,B(r,s))$-free
graph that is not an odd cycle is $(r+s-2)$-colorable. For the cases $r=2, s\geq 8$ and $r=3,s\in \{2,3\}$ they were able to 
reduce this bound to $r+s-3$.

The above results imply that every $(C_3,P_6)$-free graph is 4-colourable
(this also follows from Theorem~\ref{t-gy}).
Brandt~\cite{Br02} showed that every $(C_3,sP_2)$-free graph is $(2s-2)$-colourable for any $s\geq 3$.
This means that every $(C_3,3P_2)$-free graphs is 4-colourable.
Pyatkin~\cite{Py13} showed that every $(C_3,2P_3)$-free graph is 4-colourable,
whereas Broersma et al.~\cite{BGPS12b} showed that every $(C_3,P_2+P_4)$-free graph is 4-colourable.

\begin{open}
Determine all pairs $(H_1,H_2)$ that have the property that every $(H_1,H_2)$-free graph is $4$-colourable.
\end{open}

One problem that has had considerable attention is the classification of the computational complexity of {\sc 3-Colouring} for $(K_{1,3},H)$-free graphs.
As noted in Theorem~\ref{t-twographs}, Lozin and Purcell~\cite{LP14} showed that {\sc 3-Colouring} on $(K_{1,3},H)$-free graphs is \NP-complete  whenever 
$H\si K_{1,3}$ or $H\si \overline{2P_1+P_2}$ or $H\si C_r$ for  $r\geq 4$ or $H\si K_4$ or $H\si \Phi_{i,j}$ for $i,j\geq 0$, both even or $H\si \Phi'_i$ for $i\geq 1$ odd or $H\si \Phi''_i$ for $i\geq 0$ even.
They also observed that  
3-{\sc Colouring}  is polynomial-time solvable on $(K_{1,3},H)$-free graphs if every connected component of $H$ contains at most one triangle.  So what about the remaining cases where $H$ has a connected component containing two triangles?
Randerath, Schiermeyer and Tewes~\cite{RST02} proved that {\sc 3-Colouring} is polynomial-time solvable on $(K_{1,3},\Phi_0)$-free 
graphs, and later Kami\'nski and Lozin~\cite{KL07b} gave a linear-time algorithm.
The latter authors also showed that {\sc 3-Colouring} is polynomial-time solvable on $(K_{1,3},T_{0,0,j}^\Delta)$-free graphs for all $j\geq 0$, and Lozin and Purcell~\cite{LP14} showed that {\sc 3-Colouring} is polynomial-time solvable on $(K_{1,3},\Phi_1)$-free graphs and $(K_{1,3},\Phi_3)$-free graphs.

\begin{open}
Complete the classification of the complexity of $3$-{\sc Colouring} for $(K_{1,3},H)$-free graphs.
\end{open}

Malyshev~\cite{Ma15} characterized exactly those pairs of graphs $H_1$ and $H_2$ each with at most five vertices
for which 3-{\sc Colouring} is polynomial-time solvable when restricted to $(H_1,H_2)$-free graphs. His main result, obtained by combining known results with new results, can be summarized
as follows. 
Recall that $C_3^{++}$ is the bull. 
If at least one of $H_1$, $H_2$ is a forest and at least one of $H_1$, $H_2$ is the line graph of a forest with maximum degree at most~3 and $(H_1,H_2)\neq \{(K_{1,4},\overline{C_4+P_1}),(K_{1,4},C_3^{++})\}$,
then {\sc 3-Colouring} is polynomial-time solvable for $(H_1,H_2)$-free graphs; otherwise it is \NP-complete.

We recall that the complexity status of 4-{\sc Colouring} is still open for the class of $P_6$-free graphs. 
Randerath et al.~\cite{RST02} showed that every $(C_3,P_6)$-free graph is 4-colourable and gave a polynomial-time algorithm for finding a 4-colouring of a $(C_3,P_6)$-free graph.
The latter result also follows from Theorem~\ref{t-coloringall}~(i):4 and can be extended to $(\overline{P_1+P_3},P_6)$-free graphs
due to Theorem~\ref{t-twographs}~(ii):3.
Besides the aforementioned results that 4-{\sc Colouring} is polynomial-time solvable for $(C_5,P_6)$-free graphs (due to Theorem~\ref{t-coloringall}~(iii):15) and 
$(K_{r,s},P_6)$-free graphs for any two positive integers $r,s$ (due to Theorem~\ref{t-com2}), the following results are known as well.
Let $C_4^+$ denote the banner, which is the graph obtained from $C_4$ after adding a pendant vertex. Huang~\cite{Hu13} proved that 4-{\sc Colouring} is polynomial-time solvable for $(C_4^+,P_6)$-free graphs.  Brause, Schiermeyer, Holub, Ryj\'a\v{c}ek, Vr\'ana and Krivo\v{s}-Bellu\v{s}~\cite{SBHRVK15} proved that 4-{\sc Colouring} is polynomial-time solvable 
for $(P_6,S_{1,1,2})$-free graphs.  The same authors also considered subclasses of $(C_3^{++},P_6)$-free graphs and showed that 4-{\sc Colouring} is polynomial-time solvable for $(C_3^+,C_3^{++},P_6)$-free graphs and for
$(C_3^{++},P_6,\overline{S_{1,1,2}})$-free graphs.
Maffray and Pastor~\cite{MP15} generalized these two results by proving that 4-{\sc Colouring} is polynomial-time solvable for $(C_3^{++},P_6)$-free graphs.

\subsubsection{Certifying Algorithms}

Recall from the previous section that Ho\`ang, Moore, Recoskie, Sawada and Vatshelle~\cite{HMRSV15} showed that the number of 5-critical $P_5$-free graphs and 
the number of 5-vertex-critical $P_5$-free graphs is infinite.
They also showed that there exist 
exactly eight 5-critical $(C_5,P_5)$-free graphs.
Dhaliwal et al.~\cite{DHHMMP} proved that, for all $k\geq 1$, the number of $k$-vertex-critical $(P_5,\overline{P_5})$-free graphs is finite. They showed
that their result implies a certifying algorithm for 
{\sc $k$-Colouring} on  $(P_5,\overline{P_5})$-free graphs for all $k\geq 1$.
Randerath, Schiermeyer and Tewes~\cite{RST02} proved that the Gr\"otzsch graph is the only 4-critical $(C_3,P_6)$-free graph.
Hell and Huang~\cite{HH13} showed that, for all $k\geq 1$, the number of $k$-vertex-critical $(C_4,P_6)$-free graphs is finite. 
Moreover, they gave an explicit construction of all four $4$-vertex-critical $(C_4,P_6)$-free graphs  
and of all thirteen $5$-vertex-critical $(C_4,P_6)$-free graphs. 
Hence, they obtained certifying algorithms for  $3$-{\sc Colouring} and $4$-{\sc Colouring} on $(C_4,P_6)$-free graphs.
For all $k\geq 6$, explicit constructions of all $k$-vertex-critical graphs are unknown
(for  $k\geq 5$, no certifying algorithm is known for 
$k$-{\sc Colouring} on $(C_4,P_6)$-free graphs).
Goedgebeur and Schaudt~\cite{GS15} determined  
all 4-critical $(C_4,P_7)$-free graphs, all 4-critical $(C_5,P_7)$-free graphs and all 4-critical $(C_4,P_8)$-free graphs leading to 
certifying algorithms for {\sc 3-Colouring} restricted to these three graph classes.

\medskip
\noindent
We conclude this section by noting that,
as far as we are aware, there are no additional results for {\sc Choosability} and {\sc $k$-Choosability} known for $(H_1,H_2)$-free graphs other
than those that follow directly from previously mentioned theorems 
and two results of Esperet, Gy\'arf\'as and Maffray~\cite{EGM14} who proved that every $(K_{1,3},K_4)$-free
graph is 4-choosable and that every $(K_{1,3},K_5)$-free graph is 7-choosable.

\section{Graph Classes Defined by Other Forbidden Patterns}\label{s-strong}

In this section we consider a number of other graph classes.
We first consider
strongly $H$-free
graphs.  Recall that, given a graph $H$, the class of strongly $H$-free graphs contains those graphs that do not contain $H$ as a subgraph.

Contrast with $H$-free graphs where the graph $H$ is forbidden as an \emph{induced} subgraph:
forbidding a graph $H$ as an induced subgraph is equivalent to forbidding $H$ as a subgraph if and only if $H$ is a complete graph.
So Theorem~\ref{t-kktw} tells us that {\sc Colouring} is \NP-complete for strongly $H$-free
graphs if $H$ is a complete graph. Golovach, Paulusma and Ries~\cite{GPR12} extended this result. 
Let $T_1,\ldots,T_6$ be the trees displayed in Figure~\ref{fig:T}.
For an integer $p\geq 0$, let  $T_2^p$ be the tree obtained from $T_2$ after subdividing the edge $st$ $p$ times; note that $T_2^0=T_2$.

\begin{figure}[ht]
\centering\scalebox{0.8}{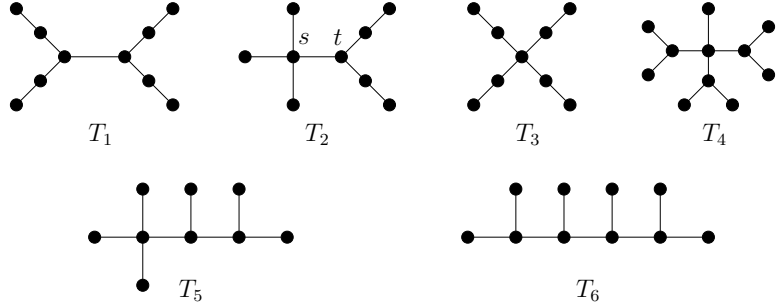}
\caption{The trees $T_1,\ldots, T_6$. 
\label{fig:T}}
\end{figure}

\begin{theorem}
\label{t-gpr12}
Let $H$ be a graph. Then the following two statements hold:
\begin{itemize}
\item [(i)]
{\sc Colouring} is polynomial-time solvable for strongly $H$-free graphs~if 
 \begin{itemize}
 \item [1.] $H$ is a forest with $\Delta(H)\leq 3$ in which each connected component has at most one vertex of degree~$3$, or
 \item [2.] $H$ is a forest  with $\Delta(H)\leq 4$ and $|V(H)|\leq 7$.\\[-8pt]
\end{itemize}
\item [(ii)] Even  $3$-{\sc Colouring} is \NP-complete for strongly $H$-free graphs if 
\begin{itemize}
\item[1.] $H$ contains a cycle, or
\item[2.] $\Delta(H)\geq 5$, or
\item[3.] $H$ has a connected component with at least two vertices of degree~$4$, or 
\item[4.] $H$ contains a subdivision of the tree $T_1$  as a subgraph, or
\item[5.] $H$ contains the tree $T_2^p$ as a subgraph for some  $0\leq p\leq 9$, or
\item[6.] $H$ contains one of the trees $T_3,T_4,T_5,T_6$ as a subgraph.
\end{itemize}
\end{itemize}
\end{theorem}

Theorems~\ref{t-kktw} and~\ref{t-gpr12} show that  {\sc Colouring} behaves differently on 
$H$-free graphs and strongly $H$-free graphs.
Theorem~\ref{t-gpr12} implies the following classification for graphs $H$ of at most seven vertices (also see~\cite{GPR12}).

\begin{theorem}\label{thm:poly}
Let $H$ be a graph. If $|V(H)|\leq 7$, then
{\sc Colouring} is polynomial-time solvable on strongly $H$-free graphs if $H$ is a forest of maximum degree at most~$4$, and \NP-complete otherwise.
\end{theorem}

The classification of {\sc Precolouring Extension} for $H$-free graphs is still open. For {\sc List Colouring}, Golovach and Paulusma~\cite{GP} gave a complete complexity classification even for graph classes defined by more than two forbidden subgraphs. 
 
\begin{theorem}
\label{t-all}
Let $\{H_1,\ldots,H_p\}$ be a finite set of graphs. Then {\sc List Colouring} is polynomial-time solvable for strongly $(H_1,\ldots,H_p)$-free graphs 
if at least one of the $H_i$, $1\leq i \leq p$, is a forest of maximum degree at most
$3$, 
every connected component of which has at most one vertex of degree~$3$.
In all other cases, even {\sc List $3$-Colouring} is \NP-complete for $(H_1,\ldots,H_p)$-free graphs.
\end{theorem}
Thus for strongly $H$-free graphs, we have the following:

\begin{open}
Complete the classification of the complexity of the problems {\sc Colouring} and {\sc Precolouring Extension} for strongly $H$-free graphs.
\end{open}

We also note that the 
classifications of the complexity of the problems {\sc $k$-Colouring} and {\sc $k$-Precolouring Extension} restricted to strongly
$H$-free graphs 
have yet to be finished. 
In particular, it would be interesting to find out whether there exists a graph $H$ such that for 
strongly 
$H$-free graphs 
3-{\sc Colouring} is polynomial-time solvable but {\sc Colouring} is \NP-complete.

We now consider graphs that are {\it $H$-minor-free}, that is, they do not contain some graph $H$ as a minor.
Robertson and Seymour showed that every class of $H$-minor-free graphs can be recognized in cubic time~\cite{RS95}.
We present some results that will allow us to determine the complexity of colouring problems on $H$-minor-free graphs.
The first is also by Robertson and Seymour~\cite{RS86}.

\begin{theorem}
\label{t-rs}
Let $H$ be any planar graph.
Then the class of $H$-minor free graphs has bounded treewidth.
\end{theorem}

The second result was proved by Jansen and Scheffler~\cite{JS97}.

\begin{theorem}\label{t-js97}
Let ${\cal G}$ be a graph class of treewidth at most $t$. 
Then {\sc List Colouring} can be solved in time $O(nk^{t+1})$ on 
a graph of ${\cal G}$ with $n$ vertices and  a $k$-list assignment.
\end{theorem}

The third and final result we need is from Garey, Johnson, and Stockmeyer~\cite{GJS74}.

\begin{theorem}
\label{t-planar}
 {\sc $3$-Colouring}  is \NP-complete for planar graphs.
\end{theorem}

In the next theorem, we present a dichotomy  for $H$-minor-free graphs. The first statement 
follows from Theorems~\ref{t-rs} and~\ref{t-js97}, and the second  from Theorem~\ref{t-planar} (after observing that the class of planar graphs is closed under taking minors).

\begin{theorem}\label{t-minor}
Let $H$ be a fixed graph. Then {\sc List Colouring} is polynomial-time solvable for $H$-minor-free graphs if $H$ is planar. 
Even {\sc $3$-Colouring} is \NP-complete for $H$-minor-free graphs if $H$ is non-planar.
\end{theorem}

Let $H$ be a graph. Then a graph is \emph{$H$-topological-minor-free} if it does not contain $H$ as a topological minor.
Grohe, Kawarabayashi, Marx and Wollan showed that every class of $H$-topological-minor-free graphs can be recognized in cubic time~\cite{GKMW11}.

By Theorem~\ref{t-planar}, and the fact that the class of planar graphs is also closed under taking topological minors, 
we see that {\sc $3$-Colouring} is \NP-complete for $H$-topological-minor-free
 graphs whenever $H$ is a non-planar graph. For every graph $H$, the class of $H$-topological-minor-free graphs is a subclass
 of the class of strongly $H$-free graphs. Hence the analogue of Theorem~\ref{t-gpr12}:(i) for $H$-topological-minor-free graphs is true.
 However, assuming \cP $\neq$ \NP, we cannot have a dichotomy equivalent to that of Theorem~\ref{t-minor}; that is, 
 the complexity of {\sc Colouring} for $H$-minor-free graphs and
 $H$-topological-minor-free graphs may be different.
 By Theorem~\ref{t-minor}, {\sc Colouring} is polynomial-time solvable for $K_{1,5}$-minor-free graphs. However, every
 graph of maximum degree at most~$4$ does not contain $K_{1,5}$ as a topological minor, and even {\sc 3-Colouring} is \NP-complete
 for graphs of maximum degree at most~$4$ according to Garey, Johnson, and Stockmeyer~\cite{GJS74}. 
 Similarly, the complexity of {\sc Colouring} for strongly $H$-free graphs and $H$-topological-minor-free graphs may be different as   Theorem~\ref{t-gpr12}~(ii):1 and the following example show.
 
 \begin{theorem}\label{t-topo}
 For all $r\geq 3$, {\sc Colouring}  is polynomial-time solvable on $C_r$-topological-minor-free graphs.
 \end{theorem}
 
 \begin{proof}
 Let  $r\geq 3$, and let $G$ be a $C_r$-topological-minor-free graph. 
 We may assume, without loss of generality, that $G$ is 2-connected.
 Suppose that $G$ contains a path $P$ on $r$ vertices.
 Because $G$ is 2-connected, there exists another path $P'$  between the end-vertices of $P$ that is internally vertex-disjoint from $P$ by Menger's Theorem.
 Then the subgraph of $G$ induced by $V(P)\cup V(P')$ contains a cycle on at least $r$ vertices. 
 Consequently, $G$ contains $C_r$ as a topological minor, which is not possible.
 Thus $G$ is strongly $P_r$-free. We apply Theorem~\ref{t-gpr12}~(i):1.\qed
 \end{proof}

 \begin{open}
Complete the classification of the complexity of {\sc Colouring}, {\sc Precolouring Extension} and {\sc List Colouring} for 
$H$-topological-minor-free graphs.
\end{open}

It remains to consider  {\sc Choosability}  restricted to the graph classes considered in this section.
Because strongly $H$-free graph classes are not closed under adding dominating vertices, we cannot just combine
Theorem~\ref{t-gpr12}~(ii) with Theorem~\ref{t-dominating} but instead need some additional results.   The first follows from a result of Bienstock, Robertson, Seymour and Thomas~\cite{BRST91}. 

\begin{theorem}
\label{t-forest}
Let $H$ be a forest with $\Delta(H)\leq 3$, in which each connected component has at most one vertex of degree~$3$.
Then every $H$-minor-free graph has pathwidth at most $|V(H)|-2$.
\end{theorem}

The next result is from Fellows et al.~\cite{FFLRSST11}.

\begin{theorem}
\label{t-chooso}
{\sc Choosability} can be solved in linear time for any graph class of bounded treewidth.
\end{theorem}

Theorems~\ref{t-forest} and~\ref{t-chooso} imply the first statement of the following theorem after observing that
a forest $H$ in which each connected component is either a path or a subdivided claw
is  a subgraph of a graph $G$ if and only if it is a minor of $G$.
The second statement follows from Theorems~\ref{t-gutner} and~\ref{t-gutnertarsi}.

\begin{theorem}\label{t-chst}
Let $H$ be a graph. Then the following two statements hold:
\begin{itemize}
\item [(i)]
{\sc Choosability} is linear-time solvable for strongly $H$-free graphs if $H$ is a forest with~$\Delta(H)\leq 3$, in which each connected component has at most one vertex of degree~$3$.
\item [(ii)] Even  $3$-{\sc Choosability} is $\Uppi_2^p$-hard for strongly $H$-free graphs if $H$ is non-planar or contains an odd cycle.
\end{itemize}
\end{theorem}
We pose the following open problem.

\begin{open}
Complete the classification of the complexity of {\sc Choosability} for strongly $H$-free  graphs.
\end{open}

When we consider $H$-minor-free graphs we obtain a full dichotomy result by using Theorem~\ref{t-rs},  Theorem~\ref{t-chooso} and  Theorem~\ref{t-gutner} and recalling that the class of planar graphs is closed under taking minors.

\begin{theorem}\label{t-minorc}
Let $H$ be a fixed graph. Then {\sc Choosability} is linear-time solvable for $H$-minor-free graphs if $H$ is planar, 
whereas even {\sc $3$-Choosability} is $\Uppi_2^p$-hard for $H$-minor-free graphs if $H$ is non-planar.
\end{theorem}

By Theorem~\ref{t-gutner} again, and the fact that the class of planar graphs is also closed under taking topological minors,
we have that {\sc $3$-Choosability} is $\Uppi_2^p$-hard for $H$-topological-minor-free graphs whenever $H$ is non-planar.
And as, for every graph $H$, the class of $H$-topological-minor-free graphs is a subclass
 of the class of strongly $H$-free graphs, the analogue of Theorem~\ref{t-chst}:(i) for $H$-topological-minor-free graphs holds.

\begin{open}
Complete the classification of the complexity of {\sc Choosability}  for $H$-topological-minor-free graphs.
\end{open}

From Theorems~\ref{t-chst} and~\ref{t-minorc}, we see that the complexity of  {\sc Choosability}  for strongly $H$-free graphs and
$H$-minor-free graphs may be 
different: for instance when $H$ is an odd cycle.
It would be interesting to determine whether there exists a graph $H$ for which the complexity of {\sc Choosability} is different for
strongly $H$-free graphs and $H$-topological-minor-free graphs, and 
whether there exists a graph $H^*$ for which the complexity of {\sc Choosability} is different for
$H^*$-minor-free graphs and $H^*$-topological-minor-free graphs.

\medskip
\noindent
{\it Acknowledgments.} We thank Konrad Dabrowski, Fran\c{c}ois Dross, Oliver Schaudt and two anonymous reviewers for helpful comments.

\section*{Appendix}
Three of the known results mentioned in our survey are not made explicit in the literature (as at the time the focus was more on $k$-{\sc colouring} and  $k$-{\sc Precolouring Extension} for $H$-free graphs than on {\sc List $k$-Colouring}).  For completeness, we give 
the proofs of these three results here.

\medskip
\noindent
The first two theorems translate statements from~\cite{BFGP13,BGPS12b} for {\sc $3$-Precolouring Extension} into statements for {\sc List 3-colouring}. 
As we show these results are implicit in~\cite{BFGP13,BGPS12b} or follow immediately from the proof methods 
used therein.

\begin{theorem}
Let $H$ be a graph. If {\sc List $3$-Colouring} is polynomial-time solvable for $H$-free graphs, then it is also polynomial-time solvable for $(P_1+H)$-free graphs.
\end{theorem}

\begin{proof}
This result can be proven by using the same arguments as the ones that Broersma et al.~\cite{BGPS12b} used for proving that $3$-{\sc Precolouring Extension}  is polynomial-time solvable.
Let $G$ be an $(H+P_1)$-free graph with a 3-list assignment $L$. If $G$ is $H$-free, we are done.
Suppose $G$ contains an induced subgraph $H'$ that is isomorphic to $H$. Because $G$ is $(H+P_1)$-free, every vertex in $V(G)\backslash V(H')$ must be adjacent to a vertex in $H'$. We guess a colouring of $V(H')$ that respects the lists.
Afterwards we apply Theorem~\ref{t-general}~(ii). Since $H'$ has a fixed size, the number of
guesses is polynomially bounded.\qed 
\end{proof}

\begin{theorem}
For every integer $s\geq 1$, {\sc List $3$-Colouring} is polynomial-time solvable on $sP_3$-free graphs.
\end{theorem}

\begin{proof}
Theorem 6 of Broersma et al.~\cite{BGPS12b} states that $3$-{\sc Precolouring Extension} can be solved in polynomial time 
on $sP_3$-graphs for any fixed $s\geq 1$.  In the proof of this theorem a polynomial-time algorithm is presented that takes as input a graph $G=(V,E)$ 
and a set of precoloured vertices $W\subseteq V$. 
We can copy the proof when the input is a graph and a $3$-list assignment after defining 
$W$ to be the set of all vertices with a list of at most 2 admissible colours. 
\qed
\end{proof}
The last result we prove in this appendix is Theorem~\ref{t-dominating}, which has been shown implicitly by Golovach et al.~\cite{GHHP12}. The proof below is only a slight adjustment of their original proof.

\medskip
\noindent
{\bf Theorem~\ref{t-dominating}.}
{\it Let ${\cal G}$ be a graph class that is closed under adding dominating vertices.
If {\sc Colouring} is \NP-hard for ${\cal G}$, then {\sc Choosability} is \NP-hard for ${\cal G}$.}

\begin{proof}
Let ${\cal G}$ be a graph class that is closed under adding dominating vertices, for which {\sc Colouring} is
 \NP-complete. Consider an instance $(G,k)$ of {\sc Colouring} where $G$ belongs to ${\cal G}$ and  $k\geq 1$ is an integer.
 We may assume without loss of generality that $\deg_G(u)\geq k$ for all $u\in V(G)$, as otherwise we add dominating vertices to $G$ and increase $k$ accordingly, in order to  obtain a pair $(G',k')$ such that $G'$ is $k'$-colourable if and only if $G$ is $k$-colourable, and by the definition of ${\cal G}$, $G'$ would belong to ${\cal G}$ as well.

We now define $k^*=k+\sum_{u\in V(G)}(\deg_G(u)-k+1)$ and
construct a graph $G^*$ 
from $G$ by adding a set of $k^*-k$ vertices $T=\{t_1,\ldots,t_{k^*-k}\}$ that are adjacent to each other and to every vertex of $G$.
By the definition of ${\cal G}$, we derive that $G^*$ belongs to ${\cal G}$.
We prove that $G$ is $k$-colourable if and only if  $G^*$ is $k^*$-choosable.

First suppose that $G^*$ is $k^*$-choosable. Then $G^*$ has a colouring $c$ that respects the list assignment ${\cal L}^*=\{L^*(u)\mid u\in V(G^*)\}$ with $L^*(u)=\{1,\ldots,k^*\}$ for all $u\in V(G^*)$. Because the $k^*-k$ vertices in $T$ are mutually adjacent, 
they are all coloured differently by $c$.
Moreover, because every vertex of $T$ is adjacent to every vertex of $G$, no vertex in $G$ has the same colour as a vertex in $T$.
Hence, by taking the restriction of $c$ to $V(G)$,  we find that $G$ is $k$-colourable. 

Now suppose that $G$ is $k$-colourable. We prove that $G^*$ is $k^*$-choosable.
In order to do this, let ${\cal L}^*=\{L^*(u)\mid u\in V(G^*)\}$ be an arbitrary $k^*$-list assignment of $G^*$.
We will construct a colouring of $G^*$ that respects ${\cal L}^*$.
We start by colouring the vertices of $T$ and, if possible, reducing $G^*$ 
by applying the following procedure:

\begin{itemize}
\item [1.] As long as there is an uncoloured vertex 
$t_j\in T$ such that  $L^*(t_j)$ contains an unused colour $x$ and there is a vertex $u\in V(G)$ with $x\notin L^*(u)$, do as follows: give $t_j$ colour $x$ and delete all vertices $u\in V(G)$ for which at least $\deg_G(u)-k+1$ used colours are not in $L^*(u)$. 
\item [2.] Afterwards, 
consider the vertices of the remaining set $T'\subseteq T$ one by one and give them any unused colour from their list.
\end{itemize}

It is possible to colour all vertices of $T$ by this procedure, because $|L^*(t_j)|=k^*$ for $j=1,\ldots,k^*-k$ and $|T|=k^*-k\leq k^*$.  
We must show that the procedure is correct.
Let $u\in V(G)$. After colouring all vertices of $T$ we can partition $T$ into two sets $A_u$ and $B_u$, where $A_u$ consists of those vertices of $T$ that received a colour not in $L^*(u)$ and $B_u=T\setminus A_u$ consists of those vertices of $T$ that received a colour from $L^*(u)$. Then the number of available colours for $u$ is $k^*-|B_u|=k^*-(|T|-|A_u|)=k^*-(k^*-k-|A_u|)=k+|A_u|$, whereas $u$ still has $\deg_G(u)$ uncoloured neighbours in $G^*$. If $k+|A_u|\geq \deg_G(u)+1$, or equivalently, if $|A_u|\geq \deg_G(u)-k+1$, then we may delete $u$; after colouring all vertices of $V(G^*)\setminus \{u\}$, we are guaranteed that there exists at least one colour in $L^*(u)$ that is not used on the neighbourhood of $u$ in $G^*$, and we can give $u$ this colour. 

After colouring the vertices in $T$ as described above, we let $U$ denote the subset of vertices of $V(G)$ that were not deleted while colouring $T$. Recall the set $T'$ defined in the procedure. We distinguish two cases.

First suppose $T'=\emptyset$. Then every $t\in T$ received a colour that does not appear in the list $L^*(u)$ for at least one vertex $u\in V(G)$ that was not yet deleted from the graph at the moment $t$ was coloured. Consequently, the size of some set $A_u$ increases by~$1$ whenever a vertex of $T$ receives a colour. Recall that a vertex $u\in U$ is deleted from the graph as soon as the size of $A_u$ reaches $\deg_G(u)-k+1$. Since $|T|= k^*-k=\sum_{u\in V(G)}(\deg_G(u)-k+1)$, every vertex of $V(G)$ is deleted from the graph at some point during the procedure. Hence $U=\emptyset$, implying that $G^*$ is $k^*$-choosable due to the correctness of our procedure. 

Now suppose $T'\neq \emptyset$ and let $t'\in T'$. Because $|L^*(t')|=k^*$ and $|T|=k^*-k$, the list $L^*(t')$ contains a set $D$ of $k$ colours that are not used as a colour for any vertex in $T$ (including $t'$ itself). We will show that $D\subseteq L^*(u)$ for every $u\in U$. For contradiction, suppose there exists a colour $y\in D$ and a vertex $w\in U$ such that $y\notin L^*(w)$. By the definition of $T'$, vertex $t'$ received a colour $z$ that appears in the list $L^*(u)$ for every $u\in U$. But according to our procedure, we would not have coloured $t'$ with colour $z$ if colour $y$ was also available; note that $y$ is not used to colour any vertex in $T\setminus \{t'\}$ by the definition of $D$. This yields the desired contradiction, implying that $D\subseteq L^*(u)$ for every $u\in U$. By symmetry of the colours, we may assume that $D=\{1,\ldots,k\}$. We assumed that $G$ is $k$-colourable, so $G$ has a colouring $c:V(G)\to \{1,\ldots,k\}$, and we can safely assign colour $c(u)$ to each $u\in U$. Due to this and the correctness of our procedure, we conclude that $G^*$ is also $k$-choosable when $T'\neq \emptyset$.
\qed
\end{proof}

\end{document}